\documentclass[journal]{IEEEtran}
\usepackage{amsmath,graphicx,cite}
\usepackage{epsfig,amssymb,amsfonts}
\usepackage{amsbsy,bbm}
\usepackage{amsmath}
\usepackage{amsthm}
\usepackage{array}
\usepackage{cite}
\usepackage{dsfont}
\usepackage[latin1]{inputenc}
\usepackage[T1]{fontenc}
\usepackage{graphicx}
\usepackage{hhline}
\usepackage{multicol,float}
\usepackage{pgfplots}
\usetikzlibrary{plotmarks}
\usepackage{tabularx}
\usepackage{tikz}
\usepackage{verbatim}
\usepackage{stmaryrd}

\usepackage{epstopdf}
\usepackage{epsfig}
\usepackage{tkz-berge}

\pgfplotsset{compat=1.14}

\newtheorem{theorem}{Theorem}
\newtheorem{lemma}[theorem]{Lemma}

\begin{document}

%
\title{Robust Model Order Selection in Large Dimensional Elliptically Symmetric Noise}


\author{\IEEEauthorblockN{Eug\'enie Terreaux\IEEEauthorrefmark{1}, Jean-Philippe Ovarlez\IEEEauthorrefmark{1}\IEEEauthorrefmark{2}, and Fr\'ed\'eric Pascal\IEEEauthorrefmark{3}} \\


\IEEEauthorblockA{\IEEEauthorrefmark{1}CentraleSup\'elec-SONDRA, 3 rue Joliot-Curie, 91190 Gif-sur-Yvette, France \\ (e-mail: eugenie.terreaux@centralesupelec.fr)} \\

\IEEEauthorblockA{\IEEEauthorrefmark{2}ONERA, DEMR/TSI, Chemin de la Huni\`ere, 91120 Palaiseau, France } \\

\IEEEauthorblockA{\IEEEauthorrefmark{1}L2S/CentraleSup\'elec-Universit\'e Paris-Sud, 3 rue Joliot-Curie, 91190 Gif-sur-Yvette, France \\ }}

\maketitle

\begin{abstract}
This paper deals with model order selection in context of correlated noise. More precisely, one considers sources embedded in an additive Complex Elliptically Symmetric (CES) noise, with unknown parameters. The main difficultly for estimating the model order lies into the noise correlation, namely the scatter matrix of the corresponding CES distribution. In this work, to tackle that problem, one adopts a two-step approach: first, we develop two different methods based on a Toeplitz-structured model for estimating this unknown scatter matrix and for whitening the correlated noise. Then, we apply Maronna's $M$-estimators on the whitened signal to estimate the covariance matrix of the ``decorrelated'' signal in order to estimate the model order. The proposed methodology is based both on robust estimation theory as well as large Random Matrix Theory, and original results are derived, proving the efficiency of this methodology. Indeed, the main theoretical contribution is to derive consistent robust estimators for the covariance matrix of the signal-plus-correlated noise in a large dimensional regime and to propose efficient methodology to estimate the rank of signal subspace. Finally, as shown in the analysis, these results show a great improvement compared to the state-of-the-art, on both simulated and real hyperspectral images. 
\end{abstract}


\IEEEpeerreviewmaketitle

\section{Introduction}

\IEEEPARstart{M}{odel} order selection is a challenging issue in signal processing for example in wireless communication \cite{Julia13}, array processing \cite{Nadler10}, or other related problems \cite{Ottersten92}, \cite{Nadler11}.
Classically, for a white noise, statistical methods such as the one based on the application of the information theoretic criteria for model order selection, allow to estimate the model order thanks to eigenvalues and eigenvectors of the covariance matrix of the signal. This is the case of the Akaike Information Criterion (AIC) \cite{Akaike74} or the Minimum Description Length (MDL) \cite{Rissanen78, schwarz78}. Other examples are the problem of source localization \cite{schmidt86}, where the estimation of the signal subspace is done by the estimation of the eigenvalues of the covariance matrix, channel identification \cite{Meraim97}, waveform estimation \cite{Liu96} and many other parametric estimation problems. Though, all these methods are no more relevant for large dimensional and correlated data. Even if particular cases have been studied for correlated signals as in \cite{Bai98} or \cite{Silverstein95}, these methods can not be generalized for all kind of signals and a whitening step, when possible, can not be systematically set up \cite{Cawse11}. 
Moreover, the commonly used statistical model for this problem has not the same matrix properties when the data are large and when they are not: the covariance matrix is not correctly apprehended and the methods fail to estimate the model order, for example in \cite{Combernoux14}, in \cite{Nadler9} or in \cite{Farsi15}. In the field of model order selection for large dimensional regime, that is when the number of snapshots $N$ and the dimension of the signal $m$ tend to infinity with a constant positive ratio, and for white or whitened noise, the Random Matrix Theory (RMT) proposes methods to estimate the model order selection relying on the study of the largest eigenvalues distribution of the covariance matrix \cite{Couillet13}. The RMT introduces new methodologies which correctly handle the statistical properties of large matrices thanks to a statistical and probability approach: see \cite{Couillet11} for a review of this theory, \cite{Kritchman9} for a general detection algorithm, \cite{Hachem13} for an adapted MUSIC detection algorithm, \cite{Pascal16} for applications to radar detection and \cite{Cawse10} for an application on hyperspectral imaging. When the noise is spatially correlated, it is still possible to estimate the model order for example by evaluating the distance between the eigenvalues of the covariance matrix \cite{Vinogradova13}. Nevertheless, these methods require a threshold that has no explicit expression and can be fastidious to obtain \cite{Terreaux15}. 
In addition to the problem of the large dimension and the correlation, another recurrent problem in signal processing is the non-Gaussianity of the noise. To be less dependent of the noise statistic, that is for the model order selection not to be degraded with a noise more or less sightly different than targeted, robust methods for model order selection have been developed \cite{Breloy16} in hyperspectral imaging \cite{Halimi16}. Nevertheless, these methods depend on unknown parameters \cite{Julia13} or are not adapted for large data. Recent results in RMT enable to correctly estimate the covariance matrix for textured signals \cite{Couillet15b}. But the correlation matrix is assume to be known and the signal is whitened before processed. \\
In this works, one considers a Complex Elliptically Symmetric (CES) noise. The CES distributions modelling is often exploited in signal processing, because of its flexibility, that is the ability to model a large panel of random signals.  The signal can be split in two parts: a texture and a speckle.  They are rather often used in various fields, as in \cite{Ovarlez11} for hyperspectral imaging, or \cite{Gini97} for radar clutter echoes modelling. This article deals with large dimansional non-Gaussian data, and proposes a robust method to estimate the model order. The robustness of our method comes from the robust estimation of the covariance matrix, with a Maronna $M$-estimator \cite{Maronna76} which assigns different weights according to the Mahalanobis distance between the signals received by the different sensors.  It is a generalization of \cite{Couillet15b} and \cite{Vinogradova14} to the case of left hand side correlation (with an unknown covariance matrix). Moreover, this article proposes a new algorithm to estimate the model order.  \\ 
In a first part, an estimator for the correlation matrix is presented: the toeplitzified Sample Covariance Matrix (SCM), that is, the SCM enforced to be of Toeplitz form \cite{Vallet14}. Indeed, as the covariance matrix is supposed to be Toeplitz, the SCM is toeplitzified as in \cite{Couillet15c} to enhance the estimation.  The data are then whitened with this Toeplitz  matrix and a robust Maronna $M$-estimator of the covariance matrix is then used after the data whitening. This robust estimation is studied and a threshold on its eigenvalues can be derived to select the model order. A second part presents the same procedure for the toeplitzified Fixed-Point (FP) estimator \cite{Tyler87} and \cite{Pascal8}. The third part presents some simulations on both simulated and real hyperspectral images. Proofs of the main results are postponed in the appendices. \\
\textit{Notations}: Matrices are in bold and capital, vectors in bold. Let $\mathbf{X}$ be a square matrix of size $s\times s$, $(\lambda)_{i}(\mathbf{X})$, $i \in \rrbracket 1,...,s \llbracket$, are the eigenvalues of $\mathbf{X}$. $Tr(\mathbf{X})$ is the trace of the matrix $\mathbf{X}$. $\left\Vert \mathbf{X} \right\Vert$ stands for the spectral norm. Let $\mathbf{A}$ be a matrix, $\mathbf{A}^T$ is the transpose of $\mathbf{A}$ and $\mathbf{A}^H$ the Hermitian transpose of $\mathbf{A}$. $\mathbf{I}_n $ is the $n \times n$ identity matrix. For any $m-$vector $\mathbf{x}$, $\mathcal{L} : \, \mathbf{x} \mapsto  \mathcal{L}(\mathbf{x})$ is the $m \times m$ matrix defined as the Toeplitz operator : $\left(  [\mathcal{L} (\mathbf{x})]_{i,j} \right) _{i\leq j} = x_{i-j}$ and $\left( [\mathcal{L} (\mathbf{x})]_{i,j} \right)_{i>j} = x_{i-j}^\ast$. For any matrix $\mathbf{A}$ of size $m \times m$, $\mathcal{T}(\mathbf{A})$ represents the matrix $\mathcal{L}(\check{\mathbf{a}})$ where $\check{\mathbf{a}}$ is a vector for which each component  $\check{\mathbf{a}}_{i ,  \, 0<i<m-1}$ contains the sum  of the $i-$th diagonal of $\mathbf{A}$ divided by $m$. For $x \in \mathds{R}$, $\delta_x$ is the Dirac measure at $x$. For any complex $z$, $z^{\star}$ is the conjugate of $z$. The notation \textit{dist} stands for the distance associated to the $L_1$ norm. $\mathrm{supp}$ is the support of a set. Eventually, $\mathcal{R}e$ and $\mathcal{I}m$ stand respectively for the real and the imaginary part for a complex number. The notation $\overset{a.s.}{\longrightarrow}$ means "tends to almost surely".

\section{Model and Assumptions}
\label{sec::2}
This section introduces the model as well as the general assumptions needed to derive the results. Let us consider the following general sources-plus-noise model. Let $\mathbf{Y} = [\mathbf{y}_0, \ldots, \mathbf{y}_{N-1}] $ be a matrix of size $m$ $\times$ $N$, containing $N$ observations $\left\{\mathbf{y}_i\right\}_{i\in \llbracket 0,N-1\rrbracket}$ of size $m$, constituted of $p$ mixed sources corrupted with an additive noise:  
\begin{equation}
\mathbf{y}_i = \displaystyle \sum_{j=1}^p s_{i,j} \, \mathbf{m}_j + \sqrt{\tau_i} \,\mathbf{C}^{1/2} \, \mathbf{x}_i\, ,  \, \hspace{0.3cm} i\in \llbracket 0,N-1 \rrbracket \, ,
\label{modele}
\end{equation}
which can be rewritten as
\begin{equation}
\mathbf{Y} = \mathbf{MS} + \mathbf{C}^{1/2}\,\mathbf{ X}\,\mathbf{ T}^{1/2} \, ,
\label{modele2}
\end{equation}

where the $\left\{\tau_i\right\}_{i\in \llbracket 0, N-1 \rrbracket}$ are positive random variables, and  $\mathbf{T}$ is the $N \times N$-diagonal matrix containing the $\left\{\tau_i\right\}_{i\in \llbracket 0, N-1 \rrbracket}$.  Moreover, the $m \times p$ matrix $\mathbf{M}$ with elements $M_{i,j} = (\mathbf{M})_{i,j} = (\mathbf{m}_j)_i $ is referred to as the mixing matrix and contains the $p$ vectors of the sources.  In this work, the additive noise is modelled thanks to the general family of Complex Elliptically Symmetric (CES) distributions \cite{Kelker70, Yao73} (see also \cite{Ollila12} for more details on CES as well as their use in signal processing). Thus, each component of the noise is characterized by a random vector $\mathbf{x}_i$ uniformly distributed on a sphere times an independent positive random scalar $\tau_i$ with unspecified  probability distribution function. The left hand side spectral correlation is handled by the scatter matrix $\mathbf C$.   

Each element  $s_{i,j}$ of the $p \times N$ matrix $\mathbf{S}$ corresponds to the power variation of each source in the received vector. This matrix can be written  $\mathbf{S} = \boldsymbol{\delta}^H \, \mathbf{\Gamma}^{1/2}$ where $\boldsymbol{\delta}$ is a $N \times p$ random matrix, independent of $\mathbf{X}$, whose elements are normally distributed with zero-mean and unit variance. $\mathbf{\Gamma}$ is a $N \times N$ Hermitian covariance matrix. 
Eventually, $\mathbf{C} = \mathcal{L} \left( [c_{0},\ldots,c_{m-1}]^T\right)$ is a $m \times m$ Hermitian nonnegative definite Toeplitz matrix:
\vspace{-0.5em}
\begin{equation*}
\mathbf{C}=
\begin{pmatrix} 
c_0 & c_1 & ... & c_{m-1} \\
c_{1}^{\star} & c_0 &... & c_{m-2} \\
 ...\\
c_{m-1}^{\star} & c_{m-2}^{\star} & ... & c_{0} 
\end{pmatrix} \, . 
\end{equation*} 
\vspace{-0.5em}

In the sequel, we will consider the following assumptions: \\
{\bf Assumption 1}: One assumes the usual random matrix regime, \textit{i.e.}:  $N \rightarrow \infty$, $m \rightarrow \infty $ and $c_N = \displaystyle \frac{m}{N} \rightarrow c > 0$, \\
\noindent {\bf Assumption 2}: for matrices $\mathbf{C}$, $\mathbf{T}$ and $\mathbf{X}$ of equation \eqref{modele2}, one has: 
\begin{itemize}
\item[$\bullet$] $\mathrm{dist}(\lambda_i(\mathbf{C}),\mathrm{supp}(\nu)) \rightarrow 0$ with $\nu$ the limit  of $ \frac{1}{N} \,\sum \delta_{\lambda_i(\mathbf{C})}$ when $N \rightarrow \infty$,
\item[$\bullet$] $\left\{c_k\right\}_{k\in \llbracket 0, m-1 \rrbracket}$ are absolutely summable coefficients, such that $c_0 \neq 0$, 
\item[$\bullet$] The random measure $\mu_N = \displaystyle \frac{1}{N} \displaystyle \sum_{i=1}^{N} \delta_{\tau_i} $ satisfies $\displaystyle\int \tau \, \mu_N(d\tau) \overset{a.s.}{\longrightarrow} 1$,
\item[$\bullet$] $\mathbf{X}$ is a white noise, with independent and  identically distributed entries with zero-mean and with unit variance, 
\item[$\bullet$] $\left\Vert \mathbf{\Gamma} \right\Vert < \infty $ and $\left\Vert \mathbf{M} \right\Vert < \infty$. 
\end{itemize}
\noindent {\bf Assumption 3}: \\
$\bullet$ In each column of $\mathbf{M}$, the coefficients are absolutely summable that is, for all fixed $j$, $\displaystyle \sum_{i=1}^m \vert M_{i,j} \vert < \infty$. This is a common assumption in several applications and especially in hyperspectral imaging.\\   
$\bullet$ $\mathbf{\Gamma}$ has coefficients absolutely summable. \\  
\noindent {\bf Assumption 4}: Let $[\mathbf{Y}]_{i,j} = [\mathbf{y}_i]_j$, then the coefficients $[\mathbf{y}_i]_j$ are absolutely summable, that is, for a fixed $i$, $\displaystyle \sum_{j} \left\vert [\mathbf{y}_i]_j \right\vert$ exists.
 \vspace{-0.5em}

\section{Model Order Selection: a Gaussian Approach}
\label{sec::3}
 
In this section, the consistency of the SCM is used to whiten the signal and to estimate the model order thanks to a Maronna $M$-estimator. The step which consists in directly evaluating  the model order with a  Maronna $M$-estimator  has been already studied in \cite{Couillet15b} for the special case of {\it spiked} model with CES white noise. In this work, one considers the more challenging problem of correlated CES  noise.    

\subsection{Whitening Step}

The noise being correlated, one proposes here to whiten it using the Toeplitz structure of the noise covariance matrix. As a reminder, the model under consideration is the following: 
\begin{equation}
\mathbf{Y} = \mathbf{M} \,  \boldsymbol{\delta}^H \, \mathbf{\Gamma}^{1/2}  + \mathbf{C}^{1/2}\,\mathbf{ X}\,\mathbf{ T}^{1/2}\, .
\label{modele3}
\end{equation}

Let $\mathbf{Y}$ be written as $\mathbf{Y} = [\mathbf{y}_0, ... , \mathbf{y}_{N-1}]$ where $\mathbf{y}_j = \left(y_{0,j}, y_{1,j}, \ldots,y_{m-1,j}\right)^T$, $j\in \llbracket 0, N-1 \rrbracket$  and let $\widetilde{\mathbf{C}}_{SCM}$ be a biased Toeplitz estimation of the covariance matrix $\mathbf{C}$ such that : 
\begin{equation}
\left[ \tilde{\mathbf{C}}_{SCM} \right]_{i,j} = \left[ \mathcal{L} (\tilde{\mathbf{c}}_{SCM})\right]_{i,j}  
\label{eq1}
\end{equation} 

with $  \tilde{c}_{SCM,k} = \frac{1}{m\,N} \, \displaystyle \sum_{i=0}^{m-1} \sum_{j=0}^{N-1} y_{i,j} \, y_{i+k,j}^\star \, \mathds{1}_{0 \leq i+k <  m}$

It can be equivalently stated as $\widehat{\mathbf{C}}_{SCM} = \displaystyle\frac{1}{N} \mathbf{Y} \, \mathbf{Y}^H$ and $\widetilde{\mathbf{C}}_{SCM} = \mathcal{T}(\widehat{\mathbf{C}}_{SCM})$. The following theorem establishes the consistency of $\widetilde{\mathbf{C}}_{SCM}$.

\begin{theorem}[Consistency of $\widetilde{\mathbf{C}}_{SCM}$] \
Under above assumptions, one has the following result:
\begin{equation}
\left\Vert   \widetilde{\mathbf{C}}_{SCM}  -  \mathds{E}[\tau] \,\mathbf{C} \right\Vert \overset{a.s.}{\longrightarrow} 0 \, . 
\label{eqth1}
\end{equation}
The covariance matrix defined by $\check{\mathbf{C}}_{SCM} = \displaystyle \frac{1}{\mathds{E}(\tau)} \, \widetilde{\mathbf{C}}_{SCM}$ characterizes the biased Toeplitz estimation of $\mathbf{C}$. 
\end{theorem}
\begin{proof}
The complete proof is in Appendix A. 
\end{proof}


This estimator is then used to whiten the samples: 
\begin{equation}
\check{\mathbf{Y}}_{wSCM} = \check{\mathbf{C}}_{SCM}^{-1/2} \,\mathbf{M} \, \boldsymbol{\delta}^H \mathbf{\Gamma}^{1/2} + \check{\mathbf{C}}_{SCM}^{-1/2}  \, \mathbf{C}^{1/2}\,\mathbf{X} \,\mathbf{ T}^{1/2}\, .
\label{scm_blanc}
\end{equation}

In practice, $\mathds{E}(\tau)$ can be empirically estimated or is supposed to be equal to $1$.


\subsection{Estimation of the covariance matrix}

Once the signal $\mathbf{Y}$ has been whitened, a robust estimation of the (unobservable) covariance matrix $\mathds{E}\left[ \mathbf{X}\,\mathbf{X}^H \right]$ can be performed through the samples $\check{\mathbf{Y}}_{wSCM}$. This estimation is said to be robust in the sense that it can annihilate the high values of the texture $\tau$, which can alter the structure quality of the estimated covariance matrix. The chosen estimator is a Maronna's $M$-estimator \cite{Maronna76}, which gives good performances for CES signals.  This robust estimation of the scatter matrix is therefore a fixed-point estimator noted $\check{\mathbf{\Sigma}}_{SCM}$ and defined through $\check{\mathbf{Y}}_{wSCM} = [\check{\mathbf{y}}_{wS0}, ...,\check{\mathbf{y}}_{wS \, N-1} ]$ as the unique solution of the following equation: 
\begin{equation}
\mathbf{\Sigma} = \frac{1}{N} \displaystyle \sum_{i=0}^{N-1} u\left(\frac{1}{m} \, \check{\mathbf{y}}_{wSi}^H \,\mathbf{\Sigma}^{-1} \,\check{\mathbf{y}}_{wSi} \right)\, \check{\mathbf{y}}_{wSi} \check{\mathbf{y}}_{wSi}^H \, ,
\label{eqsigma1}
\end{equation}
under \begin{itemize}
\item[(i)] $u$: [0, $+\infty$) $\mapsto$ (0, $+\infty$) nonnegative, continuous et non-increasing function derived thanks to the probability distribution function of the CES (for the complete calculus, see \cite{Mahot12}), 
\item[(ii)] $\phi: x \mapsto x \, u(x)$ increasing and bounded, with $ \lim\limits_{\substack{x \to \infty}} \phi(x) = \Phi_{\infty} > 1$, 
\item[(iii)] $\underset{N \longrightarrow \infty}{\text{lim}} \, c_N < \,\Phi_{\infty}^{-1}$.
\end{itemize}

Next step consists in evaluating the rank of the signal subspace from this matrix.

\subsection{Model order selection}

The mean idea  is to study the eigenvalues distribution of this Maronna $M$-estimator to find the model order or the number of sources. Indeed, in a non-RMT regime, that is if Assumption 1 is not satisfied, and in the case of a white Gaussian noise, it is possible to set a threshold such that no eigenvalues of the noise can be found upon. If eigenvalues are found beyond this threshold, they are due to sources. Here, under Assumption 1 and thanks to \cite{Bai98} in the case of a white Gaussian noise plus an additive signal, no eigenvalues outside the support of the Marchenko-Pastur law can belong to the noise. However, due to the presence of the texture matrix $\mathbf{T}$, some eigenvalues could exist upon the right edge of the Marchenko-Pastur distribution support. A more precise threshold can then be derived to ensure that no eigenvalue found upon are due to the noise. However, it does not ensure that all the sources eigenvalues will be located beyond this threshold. Indeed, this depends of the sources Signal to Noise Ratio (SNR).  \\

 The proposed estimator $\check{\mathbf{\Sigma}}_{SCM}$ has so to be analysed for CES distribution. However, some characteristics such as its eigenvalues distribution can not be easily and theoretically studied when both $m$ and $N \rightarrow \infty $ as the term $ u\left(\displaystyle \frac{1}{m} \, \mathbf{\check{y}}_{wSi}^H\, \check{\mathbf{\Sigma}}_{SCM}^{-1} \,\mathbf{\check{y}}_{wSi} \right)$ is not independent on $\mathbf{\check{y}}_{wSi}$. To fill this gap, the following white model \cite{Couillet15b} is considered : 
\begin{equation}
\mathbf{Y}_w = [\mathbf{y}_{w0},  \ldots, \mathbf{y}_{w\,N-1} ] = \mathbf{C}^{-1/2} \,\mathbf{M} \, \boldsymbol{\delta}^H \mathbf{\Gamma}^{1/2} +  \mathbf{X}\, \mathbf{T}^{1/2} \, .
\label{scm_normal}
\end{equation} 

Notice that the difference between models \eqref{scm_normal} and \eqref{scm_blanc} lies in the empirical whitening. Then, 

\begin{equation}
\hat{\mathbf{S}} \overset{\triangle}{=} \frac{1}{N} \displaystyle \sum_{i=0}^{N-1}  v\left( \tau_i \, \gamma \right) \, \mathbf{y}_{wi} \, \mathbf{y}_{wi}^{H}  \, ,
\label{matrixS}
\end{equation}

\noindent which can be rewritten as $\hat{\mathbf{S}} = \mathbf{Y}_w \,  \mathbf{D}_{\nu}  \, \mathbf{Y}_w^H$ where  $\mathbf{D}_{\nu}$ a diagonal matrix containing the $\left\{v(\tau_i \,\gamma)\right\}_i$, where:
\begin{itemize} 
\item[(i)] $g: x \mapsto \displaystyle\frac{x}{1-c \, \phi(x)}$,
\item[(ii)] $v: x \mapsto u \circ g^{-1}(x)$, $\psi: x \mapsto x\, v(x)$, with $\displaystyle\lim_{x\rightarrow \infty} \psi(x) = \displaystyle \frac{\Phi_{\infty}}{1 - c \, \Phi_{\infty}}$,
\item[(iii)] $\gamma$ is the unique solution, if defined, of the equation in $\gamma$: $1 = \displaystyle \frac{1}{N}  \, \sum_{i=1}^N \frac{\psi\left( \tau_i \, \gamma \right)}{1+ c \, \psi( \tau_i \,\gamma)}$.
\end{itemize}
Moreover, it is proved in \cite{Couillet15b} that: 
\begin{equation}
\left\Vert \hat{\mathbf{\Sigma}} - \hat{\mathbf{S}} \right\Vert \overset{a.s.}{\longrightarrow} 0\, .
\label{eqcouillet2}
\end{equation}
 
\vspace{-1em}
 
where $\hat{\mathbf{\Sigma}}$ is the unique solution (if it exists) of:
\begin{equation*}
\mathbf{\Sigma} = \frac{1}{N} \displaystyle \sum_{i=0}^{N-1}  u\left(\frac{1}{m} \,  \mathbf{y}_{wi}^H \,\mathbf{\Sigma}^{-1} \,\mathbf{y}_{wi}\right)\, \mathbf{y}_{wi} \, \mathbf{y}_{wi}^H\, .
\end{equation*}
 
\vspace{-1em} 
 
The distribution of the eigenvalues of $\hat{\mathbf{S}}$ can hence be more efficiently studied, the terms $\left[ v\left( \tau_i \, \gamma\right) \right]_{i\in \llbracket 0, \, N-1 \rrbracket}$ being independent of the  $\left\{\mathbf{x}_i\right\}_i$. The goal being to study  $\check{\mathbf{\Sigma}}_{SCM}$ which is the unique solution of \eqref{eqsigma1}, the following theorem enables to establish the relationship between $\check{\mathbf{\Sigma}}_{SCM}$ and  $\hat{\mathbf{S}}$ thanks to \eqref{eqcouillet2}.

\begin{theorem}[Convergence of $\check{\mathbf{\Sigma}}_{SCM}$] \

With previous definitions, one has the following convergence:
\begin{equation}
\left\Vert \check{\mathbf{\Sigma}}_{SCM} - \hat{\mathbf{S}} \right\Vert \overset{a.s.}{\longrightarrow} 0\, .
\label{eqcouilletmod} 
\end{equation}

\end{theorem}

\begin{proof}
The proof is provided in Appendix B. 
\end{proof}

As the eigenvalues distribution of  $\hat{\mathbf{S}}$ can be theoretically analysed when $N$, $m \rightarrow \infty$, it can characterize also those  of $\check{\mathbf{\Sigma}}_{SCM}$ thanks to \eqref{eqcouilletmod}. Under the hypothesis that there is no source present in the signal, it is possible to set a threshold similarly to \cite{Couillet15b}.  Indeed, in this case: 
\small
\begin{eqnarray}
\left\Vert \hat{\mathbf{S}} \right\Vert & =  & \left\Vert \frac{1}{N} \, \displaystyle \sum_{i=0}^{N-1} \tau_i \, v(\tau_i \, \gamma) \, \mathbf{x}_i \, \mathbf{x}_i^H \right\Vert = \left\Vert \frac{1}{N}  \displaystyle \sum_{i=0}^{N-1} \frac{1}{\gamma} \, \psi{(\tau_i \,\gamma)} \, \mathbf{x}_i \, \mathbf{x}_i^H \right\Vert \nonumber\, , \\ 
& \leq & \frac{\Phi_{\infty}}{\gamma \, \left(1-c \,\Phi_{\infty}\right)} \,  \left \Vert \, \frac{1}{N} \,  \displaystyle \sum_{i=0}^{N-1}  \mathbf{x}_i \, \mathbf{x}_i^H \right\Vert \nonumber \, .
\end{eqnarray}
\normalsize

Thanks to \cite{Bai98}, and the bounds of the Marchenko Pastur distribution support, this inequality becomes
\begin{equation}
\left\Vert  \hat{\mathbf{S}} \right\Vert  \leq t\, ,
\label{Seuil}
\end{equation} 
\noindent where the threshold $t$ is defined for the covariance matrix $\check{\mathbf{\Sigma}}_{SCM}$ by:
\begin{equation}
t = \frac{\Phi_{\infty} \, \left(1+\sqrt{c}\right)^2}{\gamma \, \left(1-c \, \Phi_{\infty}\right)}\, .
\label{eqseuil}
\end{equation}
Then, if the signal contains sources of sufficiently high SNR, eigenvalues might be found upon this threshold $t$ and all these eigenvalues correspond to sources. 
Let $\left\{ \lambda_i(\check{\mathbf{\Sigma}}_{SCM})\right\}_{i\in\llbracket 1,N\rrbracket}$ be the sorted eigenvalues of $\check{\mathbf{\Sigma}}_{SCM}$ when sources are present in the samples. As all sources are assumed to be independent,  the estimated number of sources $\hat{p}$ which corresponds to the rank of the signal subspace is then given by $\hat{p} = \displaystyle\min_k (\lambda_k > t)$, if $p << \min{(N,m)}$.

\subsection{Results}
This section is devoted to the presentation of some simulations relative to the estimation of the covariance matrix. Samples are considered here sources-free. The parameters are set to $c = 0.45$, $m=900$ and $N=2000$. Thus,  $\mathbf{Y} = \mathbf{C}^{1/2}\,\mathbf{ X}\,\mathbf{ T}^{1/2}$ with  $\mathbf{C} = \mathcal{L} \, \left(\left(\rho^0, \rho^1,\ldots ,\rho^{m-1}\right)^T\right)$ where $\rho = 0.7$ and $\mathbf{X}$ is a zero-mean complex Gaussian noise with identity covariance matrix. The texture matrix $\mathbf{T}$ is a diagonal $N\times N$-matrix containing the $\left\{\tau_i\right\}_{i\in \llbracket 0,N-1\rrbracket}$ on its diagonal where $\left\{\tau_i\right\}_i$ are i.i.d. inverse gamma distributed with mean equal to $1$ and with shape parameter equal to $10$. The function $u$ is here defined as $u : x \mapsto \displaystyle \frac{1+\alpha}{x+\alpha} $ where $\alpha$ is a fixed parameter equal to $0.1$. 

Figure \ref{fig2} shows the eigenvalues of the estimated covariance matrix $\check{\mathbf{\Sigma}}_{SCM}$ when samples $\mathbf{Y}$ have been whitened by $\check{\mathbf{C}}_{SCM}^{-1/2}$. On the figure \ref{fig3}, the signal $\mathbf{Y}$ has not been whitened. The green histogram corresponds to the eigenvalues distribution of $\hat{\mathbf{S}}$ whose histogram is expected to coincide with the distribution of the eigenvalues of $\check{\mathbf{\Sigma}}_{SCM}$ as the equation \eqref{eqcouilletmod} indicates. Moreover, the threshold  $t = \displaystyle \frac{(1+\alpha) \, (1+\sqrt{c})^2}{\gamma \, (1-c\, (1+\alpha))}$ given by \eqref{eqseuil} has been estimated and drawn in red, in order to confirm that the eigenvalues  are all smallest than the threshold.

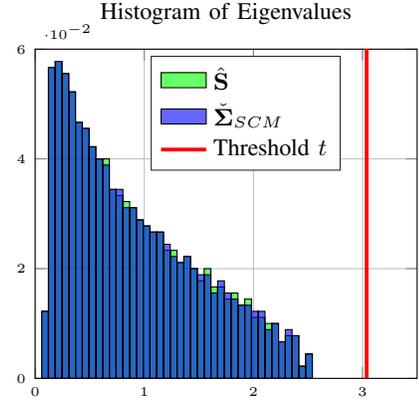
\begin{figure}
\centering 
%
%
\begin{tikzpicture}

\begin{axis}[%
width=0.28\textwidth,
at={(0.758in,0.481in)},
scale only axis,
ticklabel style = {font=\tiny},
xmin=0,
xmax=3.5,
ymin=0,
ymax=0.06,
axis background/.style={fill=white},
title={\small Histogram of Eigenvalues},
xmajorgrids,
ymajorgrids,
legend style={at={(0.80,0.98)},legend cell align=left, align=left, draw=white!15!black}
]
\addplot[ybar interval, fill=green, fill opacity=0.6, draw=black, area legend] table[row sep=crcr] {%
x	y\\
0.06	0.0122222222222222\\
0.122	0.0566666666666667\\
0.184	0.0577777777777778\\
0.246	0.0555555555555556\\
0.308	0.0522222222222222\\
0.37	0.0466666666666667\\
0.432	0.0455555555555556\\
0.494	0.0422222222222222\\
0.556	0.04\\
0.618	0.04\\
0.68	0.0344444444444444\\
0.742	0.0333333333333333\\
0.804	0.0322222222222222\\
0.866	0.0311111111111111\\
0.928	0.0288888888888889\\
0.99	0.0277777777777778\\
1.052	0.0266666666666667\\
1.114	0.0266666666666667\\
1.176	0.0233333333333333\\
1.238	0.0233333333333333\\
1.3	0.0211111111111111\\
1.362	0.0222222222222222\\
1.424	0.02\\
1.486	0.0177777777777778\\
1.548	0.02\\
1.61	0.0166666666666667\\
1.672	0.0166666666666667\\
1.734	0.0144444444444444\\
1.796	0.0155555555555556\\
1.858	0.0133333333333333\\
1.92	0.0144444444444444\\
1.982	0.0111111111111111\\
2.044	0.0111111111111111\\
2.106	0.01\\
2.168	0.01\\
2.23	0.00666666666666667\\
2.292	0.00777777777777778\\
2.354	0.00777777777777778\\
2.416	0.00222222222222222\\
2.478	0.00444444444444444\\
2.54	0.00444444444444444\\
};
\addlegendentry{\small $\hat{\mathbf{S}}$}

\addplot[ybar interval, fill=blue, fill opacity=0.6, draw=black, area legend] table[row sep=crcr] {%
x	y\\
0.06	0.0122222222222222\\
0.122	0.0566666666666667\\
0.184	0.0577777777777778\\
0.246	0.0555555555555556\\
0.308	0.0522222222222222\\
0.37	0.0466666666666667\\
0.432	0.0455555555555556\\
0.494	0.0422222222222222\\
0.556	0.04\\
0.618	0.0388888888888889\\
0.68	0.0344444444444444\\
0.742	0.0344444444444444\\
0.804	0.0311111111111111\\
0.866	0.0311111111111111\\
0.928	0.0288888888888889\\
0.99	0.0277777777777778\\
1.052	0.0266666666666667\\
1.114	0.0266666666666667\\
1.176	0.0244444444444444\\
1.238	0.0222222222222222\\
1.3	0.0211111111111111\\
1.362	0.0222222222222222\\
1.424	0.02\\
1.486	0.0188888888888889\\
1.548	0.0188888888888889\\
1.61	0.0155555555555556\\
1.672	0.0177777777777778\\
1.734	0.0155555555555556\\
1.796	0.0144444444444444\\
1.858	0.0133333333333333\\
1.92	0.0133333333333333\\
1.982	0.0122222222222222\\
2.044	0.0122222222222222\\
2.106	0.00888888888888889\\
2.168	0.01\\
2.23	0.00666666666666667\\
2.292	0.00888888888888889\\
2.354	0.00777777777777778\\
2.416	0.00222222222222222\\
2.478	0.00444444444444444\\
2.54	0.00444444444444444\\
};
\addlegendentry{\small $\check{\mathbf{\Sigma}}_{SCM}$}

\addplot [color=red, line width=0.5mm]
  table[row sep=crcr]{%
3.03743394491305	0\\
3.03743394491305	0.2\\
};
\addlegendentry{\small Threshold $t$}

\end{axis}
\end{tikzpicture}%
	\caption{Eigenvalues of the covariance matrices $\check{\mathbf{\Sigma}}_{SCM}$ and $\hat{\mathbf{S}}$ when the signal $\mathbf{Y}$ has been whitened by  $\check{\mathbf{C}}_{SCM}$ and 		the corresponding threshold $t$ ($\rho = 0.7$, $m=900$, $N=2000$, $\tau = $ inverse gamma, $\alpha = 0.1$).}
	\label{fig2}
\end{figure}

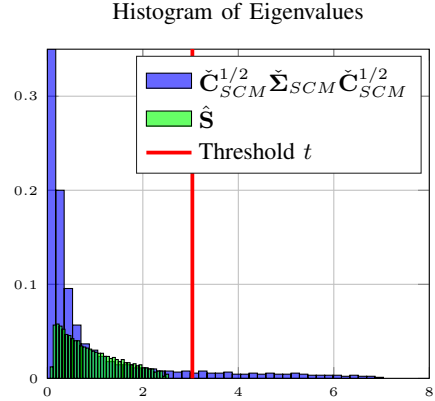
\begin{figure}
\centering 
%
%
\begin{tikzpicture}

\begin{axis}[%
width=0.28\textwidth,
at={(0.758in,0.481in)},
scale only axis,
ticklabel style = {font=\tiny},
xmin=0,
xmax=8,
ymin=0,
ymax=0.35,
axis background/.style={fill=white},
title={\small Histogram of Eigenvalues},
xmajorgrids,
ymajorgrids,
legend style={at={(0.98,0.98)}, legend cell align=left, align=left, draw=white!15!black}
]
\addplot[ybar interval, fill=blue, fill opacity=0.6, draw=black, area legend] table[row sep=crcr] {%
x	y\\
0	0.35\\
0.176	0.2\\
0.352	0.0955555555555556\\
0.528	0.0566666666666667\\
0.704	0.0366666666666667\\
0.88	0.03\\
1.056	0.0233333333333333\\
1.232	0.0177777777777778\\
1.408	0.0144444444444444\\
1.584	0.0144444444444444\\
1.76	0.0111111111111111\\
1.936	0.0111111111111111\\
2.112	0.01\\
2.288	0.00777777777777778\\
2.464	0.00777777777777778\\
2.64	0.00666666666666667\\
2.816	0.00777777777777778\\
2.992	0.00555555555555556\\
3.168	0.00777777777777778\\
3.344	0.00555555555555556\\
3.52	0.00555555555555556\\
3.696	0.00666666666666667\\
3.872	0.00444444444444444\\
4.048	0.00444444444444444\\
4.224	0.00555555555555556\\
4.4	0.00444444444444444\\
4.576	0.00555555555555556\\
4.752	0.00444444444444444\\
4.928	0.00444444444444444\\
5.104	0.00555555555555556\\
5.28	0.00444444444444444\\
5.456	0.00333333333333333\\
5.632	0.00333333333333333\\
5.808	0.00333333333333333\\
5.984	0.00333333333333333\\
6.16	0.00222222222222222\\
6.336	0.00333333333333333\\
6.512	0.00222222222222222\\
6.688	0.00222222222222222\\
6.864	0.00111111111111111\\
7.04	0.00111111111111111\\
};
\addlegendentry{\small  $\check{\mathbf{C}}_{SCM}^{1/2} \check{\mathbf{\Sigma}}_{SCM} \check{\mathbf{C}}_{SCM}^{1/2}$}

\addplot[ybar interval, fill=green, fill opacity=0.6, draw=black, area legend] table[row sep=crcr] {%
x	y\\
0.06	0.0122222222222222\\
0.122	0.0566666666666667\\
0.184	0.0577777777777778\\
0.246	0.0555555555555556\\
0.308	0.0522222222222222\\
0.37	0.0466666666666667\\
0.432	0.0455555555555556\\
0.494	0.0422222222222222\\
0.556	0.04\\
0.618	0.04\\
0.68	0.0344444444444444\\
0.742	0.0333333333333333\\
0.804	0.0322222222222222\\
0.866	0.0311111111111111\\
0.928	0.0288888888888889\\
0.99	0.0277777777777778\\
1.052	0.0266666666666667\\
1.114	0.0266666666666667\\
1.176	0.0233333333333333\\
1.238	0.0233333333333333\\
1.3	0.0211111111111111\\
1.362	0.0222222222222222\\
1.424	0.02\\
1.486	0.0177777777777778\\
1.548	0.02\\
1.61	0.0166666666666667\\
1.672	0.0166666666666667\\
1.734	0.0144444444444444\\
1.796	0.0155555555555556\\
1.858	0.0133333333333333\\
1.92	0.0144444444444444\\
1.982	0.0111111111111111\\
2.044	0.0111111111111111\\
2.106	0.01\\
2.168	0.01\\
2.23	0.00666666666666667\\
2.292	0.00777777777777778\\
2.354	0.00777777777777778\\
2.416	0.00222222222222222\\
2.478	0.00444444444444444\\
2.54	0.00444444444444444\\
};
\addlegendentry{\small $\hat{\mathbf{S}}$}

\addplot [color=red, line width=0.5mm]
  table[row sep=crcr]{%
3.03743394491305	0\\
3.03743394491305	2\\
};
\addlegendentry{\small Threshold $t$}

\end{axis}
\end{tikzpicture}%
	\caption{Eigenvalues of the covariance matrices $\check{\mathbf{\Sigma}}_{SCM}$ and $\hat{\mathbf{S}}$ when the signal $\mathbf{Y}$ has not been whitened by  $\check{\mathbf{C}}_{SCM}$ 		and the corresponding threshold $t$ ($\rho = 0.7$, $m=900$, $N=2000$, $\tau = $ inverse gamma, $\alpha = 0.1$).}	
	\label{fig3}
\end{figure}

As the eigenvalues distribution of $\check{\mathbf{\Sigma}}_{SCM}$  are closed to those of $\hat{\mathbf{S}}$, the fixed-point estimator correctly annihilates the influence of the textures $\tau_i$'s and the whitening balances the matrix of correlation. On Figure \ref{fig2}, we can observe  that the eigenvalues do not exceed the upper bound $t$. When the signal has not been whitened, this threshold $t$ does not theoretically correspond. Indeed, in Figure \ref{fig3}, the threshold is found to be smaller than the largest eigenvalues of the estimated covariance matrix. These figures illustrate first the results of Theorem 2 and show the importance of the whitening process.  \\

Figure \ref{fig4} presents the eigenvalues distributions of $\hat{\mathbf{S}}$ and $\check{\mathbf{C}}_{SCM}$ for samples distributed according to a different CES distribution. Here, the texture $\mathbf{T}$ is a diagonal matrix containing the $\left\{\tau_i\right\}_{i\in \llbracket 0,, N-1\rrbracket}$ on its diagonal where each $\tau_i$  is independent and identically distributed and follows a distribution equal to $t^2$ where $t$ is a Student-t distributed random variable with parameter $100$ and $\alpha = 0.1$. The eigenvalues are not so close than the eigenvalues of $\hat{\mathbf{S}}$ and are found to get closer to the threshold $t$. If the distribution of $\tau$ is getting away to the one for which the function $u$ has been calculated, the method seems so to be less reliable. To fill this gap, we propose  to enhance the proposed SCM-based method for the whitening through robust $M$-estimators-based method. 

{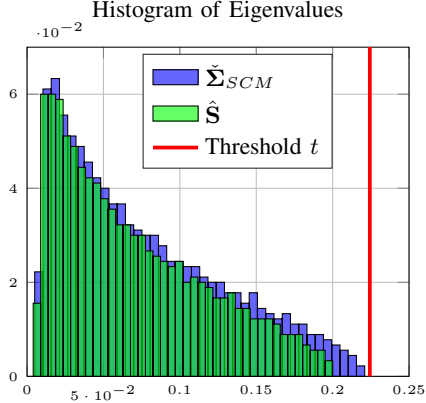
\begin{figure}[H]
\centering 
%
%
\begin{tikzpicture}

\begin{axis}[%
width=0.28\textwidth,
at={(0.758in,0.481in)},
scale only axis,
ticklabel style = {font=\tiny},
xmin=0,
xmax=0.25,
ymin=0,
ymax=0.07,
axis background/.style={fill=white},
title={\small Histogram of Eigenvalues},
xmajorgrids,
ymajorgrids,
legend style={at={(0.80,0.98)}, legend cell align=left, align=left, draw=white!15!black}
]

\addplot[ybar interval, fill=blue, fill opacity=0.6, draw=black, area legend] table[row sep=crcr] {%
x	y\\
0.005	0.0222222222222222\\
0.0104	0.0611111111111111\\
0.0158	0.0633333333333333\\
0.0212	0.0555555555555556\\
0.0266	0.0511111111111111\\
0.032	0.0488888888888889\\
0.0374	0.0455555555555556\\
0.0428	0.0422222222222222\\
0.0482	0.04\\
0.0536	0.0366666666666667\\
0.059	0.0366666666666667\\
0.0644	0.0322222222222222\\
0.0698	0.0311111111111111\\
0.0752	0.03\\
0.0806	0.03\\
0.086	0.0277777777777778\\
0.0914	0.0244444444444444\\
0.0968	0.0244444444444444\\
0.1022	0.0233333333333333\\
0.1076	0.0233333333333333\\
0.113	0.0211111111111111\\
0.1184	0.02\\
0.1238	0.02\\
0.1292	0.0177777777777778\\
0.1346	0.0177777777777778\\
0.14	0.0155555555555556\\
0.1454	0.0177777777777778\\
0.1508	0.0144444444444444\\
0.1562	0.0133333333333333\\
0.1616	0.0122222222222222\\
0.167	0.0133333333333333\\
0.1724	0.0111111111111111\\
0.1778	0.0111111111111111\\
0.1832	0.00888888888888889\\
0.1886	0.00888888888888889\\
0.194	0.00777777777777778\\
0.1994	0.00666666666666667\\
0.2048	0.00555555555555556\\
0.2102	0.00444444444444444\\
0.2156	0.00222222222222222\\
0.221	0.00222222222222222\\
};
\addlegendentry{\small $\check{\mathbf{\Sigma}}_{SCM} $}

\addplot[ybar interval, fill=green, fill opacity=0.6, draw=black, area legend] table[row sep=crcr] {%
x	y\\
0.004	0.0155555555555556\\
0.0089	0.06\\
0.0138	0.06\\
0.0187	0.0588888888888889\\
0.0236	0.0511111111111111\\
0.0285	0.0488888888888889\\
0.0334	0.0444444444444444\\
0.0383	0.0422222222222222\\
0.0432	0.0411111111111111\\
0.0481	0.0377777777777778\\
0.053	0.0355555555555556\\
0.0579	0.0322222222222222\\
0.0628	0.0322222222222222\\
0.0677	0.03\\
0.0726	0.03\\
0.0775	0.0266666666666667\\
0.0824	0.0255555555555556\\
0.0873	0.0244444444444444\\
0.0922	0.0233333333333333\\
0.0971	0.0244444444444444\\
0.102	0.02\\
0.1069	0.0211111111111111\\
0.1118	0.02\\
0.1167	0.0188888888888889\\
0.1216	0.0166666666666667\\
0.1265	0.0166666666666667\\
0.1314	0.0177777777777778\\
0.1363	0.0144444444444444\\
0.1412	0.0144444444444444\\
0.1461	0.0122222222222222\\
0.151	0.0122222222222222\\
0.1559	0.0122222222222222\\
0.1608	0.0111111111111111\\
0.1657	0.00888888888888889\\
0.1706	0.00888888888888889\\
0.1755	0.00888888888888889\\
0.1804	0.00666666666666667\\
0.1853	0.00555555555555556\\
0.1902	0.00555555555555556\\
0.1951	0.00333333333333333\\
0.2	0.00333333333333333\\
};
\addlegendentry{\small $\hat{\mathbf{S}}$}

\addplot [color=red, line width=0.5mm]
  table[row sep=crcr]{%
0.224345142665506	0\\
0.224345142665506   2\\
};
\addlegendentry{\small Threshold $t$}

\end{axis}
\end{tikzpicture}%
	\caption{Eigenvalues of the covariance matrices $\check{\mathbf{\Sigma}}_{SCM}$ and $\hat{\mathbf{S}}$ when the signal is whitened by  $\check{\mathbf{C}}_{SCM}$ and the calculated threshold ($\rho = 0.7$, $m=900$, $N=2000$, $\tau = $ $t^2$, $\alpha = 0.1$).}
	\label{fig4}
\end{figure}}


\section{Model Order Selection: a robust method approach}

This section aims at developing a robust estimator based technique to whiten the signal instead of the previous SCM-based one. This section follows the same steps than in the previous section but by using a $M$-estimator in the whitening process.

\subsection{Whitening Step}

Let $\widetilde{\mathbf{C}}_{FP}$ be an biased estimator of the covariance matrix $\mathbf{C}$ such that $ \widetilde{\mathbf{C}}_{FP}  = \mathcal{T} (\widehat{\mathbf{C}}_{FP})  $  where $\widehat{\mathbf{C}}_{FP}$ is the unique solution to the Maronna's $M$-estimator \cite{Maronna76}:  
\begin{equation*}
\mathbf{Z} = \frac{1}{N} \, \displaystyle \sum_{i=0}^{N-1} u\left(\frac{1}{m} \,\mathbf{y}_i^{H} \, \mathbf{Z}^{-1} \, \mathbf{y}_i \right) \,  \mathbf{y}_i \, \mathbf{y}_i^{H}\, .
\end{equation*} 
\noindent As in the previous section, $u(.)$  is a function derived thanks to the probability distribution function of the CES noise: $u$: [0, $+\infty$) $\mapsto$ (0, $+\infty$) nonnegative, continuous and non-increasing. The following theorem stands for $\widetilde{\mathbf{C}}_{FP}$:

\begin{theorem}[Consistency of $\widetilde{\mathbf{C}}_{FP}$]  Let $\widetilde{\mathbf{C}}_{FP}$ be a fixed-point estimator of the covariance matrix $\mathbf{C}$ as defined above, the following result holds: 
\begin{equation}
\left\Vert \widetilde{\mathbf{C}}_{FP} -  \mathds{E}\left[v(\tau \,\gamma)\,\tau\right] \,\mathbf{C} \right\Vert \overset{a.s.}{\longrightarrow} 0 \, ,
\label{eqth2}
\end{equation}
\noindent where:
\begin{itemize}
\item[i)] $\phi: x \mapsto x \,u(x)$ increasing and bounded, with $ \lim\limits_{\substack{x \to \infty}} \phi(x) = \Phi_{\infty} > 1$,
\item[ii)] $\underset{N \rightarrow \infty}{\text{lim}} \, c_N < \, \Phi_{\infty}^{-1}$, 
\item[iii)] $g: x \mapsto \displaystyle \frac{x}{1-c \,\phi(x)}$,
\item[iv)] $v: x \mapsto u \circ g^{-1}(x)$, $\psi: x \mapsto x\, v(x)$, 
\item[v)] $\gamma$ is the unique solution, if defined, of: $1 = \displaystyle \frac{1}{N}  \, \sum_{i=1}^N \frac{\psi\left( \tau_i \, \gamma \right)}{1+ c \, \psi( \tau_i \,\gamma)}$.
\end{itemize}
The covariance matrix $\check{\mathbf{C}}_{FP} =\displaystyle  \frac{ \tilde{\mathbf{C}}_{FP}  }{\mathds{E} \left[v(\tau \,\gamma) \, \tau\right]}$ characterizes the  estimator of the true covariance matrix $\mathbf{C}$. 
\end{theorem}

\begin{proof}
The proof, inspired by \cite{Vinogradova14} and \cite{Loubaton16}, is provided in Appendix B. 
\end{proof}

\textit{Remark}: 
\cite{Couillet13} proves that  $\tilde{\mathbf{C}}_{SCM} = \phi^{-1}(1) \,\tilde{\mathbf{C}}_{FP}$. When the function $u$ is well chosen, it is possible to have $\phi^{-1}(1) = 1$ and $\tilde{\mathbf{C}}_{SCM} = \tilde{\mathbf{C}}_{FP} $, as it will be the case for the $u$ chosen in the following sections. But even in this case, $\check{\mathbf{C}}_{SCM}$ and $\check{\mathbf{C}}_{FP} $ differ up to a scale factor as $\check{\mathbf{C}}_{SCM} = \displaystyle \frac{\mathds{E} \left[v(\tau \gamma) \,\tau\right]}{\mathds{E} (\tau)} \, \check{\mathbf{C}}_{FP} $. \\




 As in the previous section,  the samples $\mathbf{Y}$ can then be whitened thanks to $ \check{\mathbf{C}}_{FP}^{-1/2}$.  Let $\check{\mathbf{Y}}_{wFP} = [\check{\mathbf{y}}_{wF0} ,... \check{\mathbf{y}}_{wF \, N-1}]$ be the whitened samples: 
\begin{equation*}
\check{\mathbf{Y}}_{wFP} = \check{\mathbf{C}}_{FP}^{-1/2} \,\mathbf{M} \, \boldsymbol{\delta}^H \, \mathbf{\Gamma}^{1/2} + \check{\mathbf{C}}_{FP}^{-1/2}  \,  \mathbf{C}^{1/2}\, \mathbf{X}\, \mathbf{T}^{1/2}\, .
\end{equation*}

The parameter $\mathds{E}[\tau]$ can be in practice evaluated with the empirical estimator of the mean, or, as in the previous section, be considered as equal to one. The quantity $\mathds{E}[v(\tau \,\gamma)\,\tau]$ can be also evaluated through  an estimate $\hat{\gamma}$ of $\gamma$ as explained in the Results section.

\subsection{Robust estimation of the covariance matrix and model order selection}

The robust estimation of the covariance matrix and the model order selection are done as previously.
The robust estimator of the scatter matrix of the whitened signal $\check{\mathbf{Y}}_{wFP}$ is a fixed-point estimator denoted by $\check{\mathbf{\Sigma}}_{FP}$ and defined as the unique solution of the equation: 
\begin{equation}
\mathbf{\Sigma} = \frac{1}{N} \displaystyle \sum_{i=0}^{N-1} u\left(\frac{1}{m} \, \check{\mathbf{y}}_{wFi}^H \,\mathbf{\Sigma}^{-1} \,\check{\mathbf{y}}_{wFi} \right)\, \check{\mathbf{y}}_{wFi} \, \check{\mathbf{y}}_{wFi}^H\, .
\label{eqsigma2}
\end{equation}

\noindent Thus, $\check{\mathbf{\Sigma}}_{FP}$ is a robust estimator of the covariance matrix of the whitened signal. \\

The equation \eqref{eqcouilletmod} is still effective when replacing $\check{\mathbf{\Sigma}}_{SCM}$ by $\check{\mathbf{\Sigma}}_{FP}$. Indeed, Theorem 2 can be adapted as follows: 
\begin{theorem} The following convergence holds:
\label{theorem2}
\begin{equation}
\left\Vert \check{\mathbf{\Sigma}}_{FP} - \hat{\mathbf{S}}  \right\Vert \overset{a.s.}{\longrightarrow} 0 \, .
\label{eqth2mod}
\end{equation} 
\end{theorem}
\begin{proof}
The proof is the same as in Theorem 2 and is provided in Appendix B.  
\end{proof}

The same threshold $t$ given by equation \eqref{eqseuil} can be used on the eigenvalues of $\check{\mathbf{\Sigma}}_{FP}$ to estimate $p$. The final corresponding algorithm is presented below. \\

\subsection{Results}

As in the previous section, it seems interesting to analyse the eigenvalues distributions of $\hat{\mathbf{S}}$ and $\check{\mathbf{\Sigma}}_{FP}$. For the next simulations, source-free samples are considered and the parameters are set to $c = 0.45$, $m=900$ and $N=2000$. The function $u$ chosen for the FP and Maronna $M$-estimators is the same function as before with $\alpha=0.1$. \\
Figure \ref{fig6} presents the eigenvalues distribution of the covariance matrices $\hat{\mathbf{S}}$ and $\check{\mathbf{\Sigma}}_{FP}$ when the signal has been whitened by $\check{\mathbf{C}}_{FP}$. One can notice that the results are the same as Figure \ref{fig2}:  for $N$ large, the distribution of eigenvalues is almost the same as those of $\mathbf{\hat{S}}$. However, as the rate of convergence of \eqref{eqth2} is faster than in  \eqref{eqth1}, it is more interesting to consider the robust method. Moreover, if a robust estimator is not used after the whitening process, the eigenvalues distribution will not follow those of $\hat{\mathbf{S}}$ and will exceed  the threshold $t$. 

\begin{figure}
\centering 
%
%
\begin{tikzpicture}

\begin{axis}[%
width=0.28\textwidth,
at={(0.758in,0.481in)},
scale only axis,
ticklabel style = {font=\tiny},
xmin=0,
xmax=3.5,
ymin=0,
ymax=0.06,
axis background/.style={fill=white},
title={\small Histogram of eigenvalues},
xmajorgrids,
ymajorgrids,
legend style={at={(0.80,0.98)},legend cell align=left, align=left, draw=white!15!black}
]
\addplot[ybar interval, fill=green, fill opacity=0.6, draw=black, area legend] table[row sep=crcr] {%
x	y\\
0.06	0.0122222222222222\\
0.122	0.0566666666666667\\
0.184	0.0577777777777778\\
0.246	0.0555555555555556\\
0.308	0.0522222222222222\\
0.37	0.0466666666666667\\
0.432	0.0455555555555556\\
0.494	0.0422222222222222\\
0.556	0.04\\
0.618	0.04\\
0.68	0.0344444444444444\\
0.742	0.0333333333333333\\
0.804	0.0322222222222222\\
0.866	0.0311111111111111\\
0.928	0.0288888888888889\\
0.99	0.0277777777777778\\
1.052	0.0266666666666667\\
1.114	0.0266666666666667\\
1.176	0.0233333333333333\\
1.238	0.0233333333333333\\
1.3	0.0211111111111111\\
1.362	0.0222222222222222\\
1.424	0.02\\
1.486	0.0177777777777778\\
1.548	0.02\\
1.61	0.0166666666666667\\
1.672	0.0166666666666667\\
1.734	0.0144444444444444\\
1.796	0.0155555555555556\\
1.858	0.0133333333333333\\
1.92	0.0144444444444444\\
1.982	0.0111111111111111\\
2.044	0.0111111111111111\\
2.106	0.01\\
2.168	0.01\\
2.23	0.00666666666666667\\
2.292	0.00777777777777778\\
2.354	0.00777777777777778\\
2.416	0.00222222222222222\\
2.478	0.00444444444444444\\
2.54	0.00444444444444444\\
};
\addlegendentry{\small $\hat{\mathbf{S}}$}

\addplot[ybar interval, fill=blue, fill opacity=0.6, draw=black, area legend] table[row sep=crcr] {%
x	y\\
0.06	0.0122222222222222\\
0.122	0.0566666666666667\\
0.184	0.0577777777777778\\
0.246	0.0555555555555556\\
0.308	0.0522222222222222\\
0.37	0.0477777777777778\\
0.432	0.0455555555555556\\
0.494	0.0422222222222222\\
0.556	0.04\\
0.618	0.0388888888888889\\
0.68	0.0355555555555556\\
0.742	0.0333333333333333\\
0.804	0.0322222222222222\\
0.866	0.03\\
0.928	0.03\\
0.99	0.0277777777777778\\
1.052	0.0255555555555556\\
1.114	0.0266666666666667\\
1.176	0.0244444444444444\\
1.238	0.0233333333333333\\
1.3	0.0211111111111111\\
1.362	0.0222222222222222\\
1.424	0.0188888888888889\\
1.486	0.0188888888888889\\
1.548	0.0188888888888889\\
1.61	0.0166666666666667\\
1.672	0.0177777777777778\\
1.734	0.0144444444444444\\
1.796	0.0144444444444444\\
1.858	0.0133333333333333\\
1.92	0.0144444444444444\\
1.982	0.0122222222222222\\
2.044	0.0111111111111111\\
2.106	0.00888888888888889\\
2.168	0.01\\
2.23	0.00777777777777778\\
2.292	0.00777777777777778\\
2.354	0.00666666666666667\\
2.416	0.00333333333333333\\
2.478	0.00333333333333333\\
2.54	0.00333333333333333\\
};
\addlegendentry{\small $\check{\mathbf{\Sigma}}_{FP}$}

\addplot [color=red, line width=0.5mm]
  table[row sep=crcr]{%
3.03743394491305	0\\
3.03743394491305	0.2\\
};
\addlegendentry{\small Threshold $t$}

\end{axis}
\end{tikzpicture}%
	\caption{Eigenvalues of the covariance matrices $\hat{\mathbf{S}}$ and $\check{\mathbf{\Sigma}}_{FP}$ when the signal is whitened through $\check{\mathbf{C}}_{FP}$ and the corresponding threshold $t$ ($\rho = 0.7$, $m=900$, $N=2000$, $\tau \sim $ inverse gamma).}
	\label{fig6}
\end{figure}
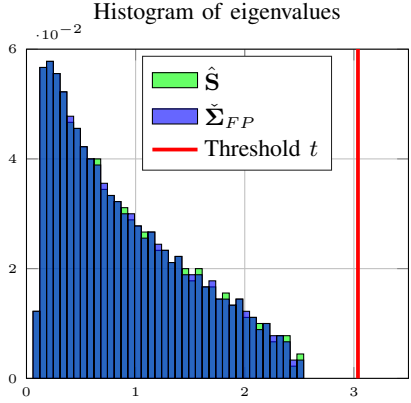
	

For robustness analysis (not the same texture distribution for the $u$ function and the observed samples), Figure \ref{fig8} shows quite good results when $\mathbf{T}$ is a diagonal matrix containing the $\{\tau_i\}_{i\in \llbracket0,N-1\rrbracket}$ on its diagonal where $\{\tau_i\}_{i\in \llbracket 0, N-1\rrbracket}$ are i.i.d. and follow a distribution equal to $t^2$ with $t$ a Student-t random variable with parameter $100$ and $\alpha = 0.1$. 

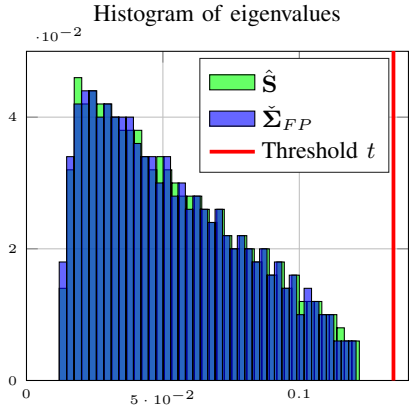
\begin{figure}
\centering 
%
%
\begin{tikzpicture}

\begin{axis}[%
width=0.28\textwidth,
at={(0.758in,0.481in)},
scale only axis,
ticklabel style = {font=\tiny},
xmin=0,
xmax=0.14,
ymin=0,
ymax=0.05,
axis background/.style={fill=white},
title={\small Histogram of eigenvalues},
xmajorgrids,
ymajorgrids,
legend style={at={(0.95,0.98)},legend cell align=left, align=left, draw=white!15!black}
]
\addplot[ybar interval, fill=green, fill opacity=0.6, draw=black, area legend] table[row sep=crcr] {%
x	y\\
0.012	0.014\\
0.01475	0.032\\
0.0175	0.046\\
0.02025	0.042\\
0.023	0.044\\
0.02575	0.042\\
0.0285	0.042\\
0.03125	0.04\\
0.034	0.038\\
0.03675	0.038\\
0.0395	0.038\\
0.04225	0.034\\
0.045	0.032\\
0.04775	0.034\\
0.0505	0.032\\
0.05325	0.03\\
0.056	0.028\\
0.05875	0.028\\
0.0615	0.028\\
0.06425	0.026\\
0.067	0.024\\
0.06975	0.026\\
0.0725	0.022\\
0.07525	0.02\\
0.078	0.022\\
0.08075	0.02\\
0.0835	0.018\\
0.08625	0.02\\
0.089	0.016\\
0.09175	0.018\\
0.0945	0.014\\
0.09725	0.016\\
0.1	0.012\\
0.10275	0.012\\
0.1055	0.012\\
0.10825	0.01\\
0.111	0.01\\
0.11375	0.008\\
0.1165	0.006\\
0.11925	0.006\\
0.122	0.006\\
};
\addlegendentry{\small $\hat{\mathbf{S}}$}

\addplot[ybar interval, fill=blue, fill opacity=0.6, draw=black, area legend] table[row sep=crcr] {%
x	y\\
0.012	0.018\\
0.01472	0.034\\
0.01744	0.042\\
0.02016	0.044\\
0.02288	0.044\\
0.0256	0.04\\
0.02832	0.042\\
0.03104	0.04\\
0.03376	0.04\\
0.03648	0.04\\
0.0392	0.036\\
0.04192	0.034\\
0.04464	0.034\\
0.04736	0.03\\
0.05008	0.034\\
0.0528	0.028\\
0.05552	0.03\\
0.05824	0.026\\
0.06096	0.028\\
0.06368	0.026\\
0.0664	0.024\\
0.06912	0.026\\
0.07184	0.022\\
0.07456	0.02\\
0.07728	0.022\\
0.08	0.02\\
0.08272	0.018\\
0.08544	0.02\\
0.08816	0.016\\
0.09088	0.018\\
0.0936	0.014\\
0.09632	0.016\\
0.09904	0.01\\
0.10176	0.014\\
0.10448	0.012\\
0.1072	0.01\\
0.10992	0.01\\
0.11264	0.006\\
0.11536	0.006\\
0.11808	0.006\\
0.1208	0.006\\
};
\addlegendentry{\small $\check{\mathbf{\Sigma}}_{FP}$}

\addplot [color=red, line width=0.5mm]
  table[row sep=crcr]{%
0.13450992280155	0\\
0.13450992280155	0.2\\
};
\addlegendentry{\small Threshold $t$}

\end{axis}
\end{tikzpicture}%
	\caption{Eigenvalues of the covariance matrices $\hat{\mathbf{S}}$ and $\check{\mathbf{\Sigma}}_{FP}$  when the signal is whitened through  $\check{\mathbf{C}}_{FP}$ and the corresponding threshold ($\rho = 0.7$, $m=900$, $N=2000$, $\tau = t^2$, $t \sim $ student)}
	\label{fig8}
\end{figure}

%

Figure \ref{fig9_bis} presents the same histograms as in Figure \ref{fig6} for a single source of SNR equal to $10~$dB present in the samples. One can observe that only single eigenvalue exceeds the threshold and that the noise eigenvalues distribution of $\check{\mathbf{\Sigma}}_{FP}$ fits well those of $\hat{\mathbf{S}}$. 

\begin{figure}[htbp]
	\centering 
%
%
\begin{tikzpicture}

\begin{axis}[%
width=0.28\textwidth,
at={(0.758in,0.481in)},
scale only axis,
xmin=0,
xmax=3.5,
ymin=0,
ymax=0.1,
axis background/.style={fill=white},
title={\small Histogram of eigenvalues},
xmajorgrids,
ymajorgrids,
legend style={at={(0.85,0.98)},legend cell align=left, align=left, draw=white!15!black}
]
\addplot[ybar interval, fill=green, fill opacity=0.6, draw=black, area legend] table[row sep=crcr] {%
x	y\\
0	0.00111111111111111\\
0.1	0.0811111111111111\\
0.2	0.0911111111111111\\
0.3	0.0811111111111111\\
0.4	0.0733333333333333\\
0.5	0.0655555555555556\\
0.6	0.0622222222222222\\
0.7	0.0555555555555556\\
0.8	0.0511111111111111\\
0.9	0.0466666666666667\\
1	0.0455555555555556\\
1.1	0.04\\
1.2	0.0377777777777778\\
1.3	0.0344444444444444\\
1.4	0.0311111111111111\\
1.5	0.0311111111111111\\
1.6	0.0288888888888889\\
1.7	0.0244444444444444\\
1.8	0.0244444444444444\\
1.9	0.0211111111111111\\
2	0.0188888888888889\\
2.1	0.0166666666666667\\
2.2	0.0133333333333333\\
2.3	0.0122222222222222\\
2.4	0.00777777777777778\\
2.5	0.00333333333333333\\
2.6	0.00333333333333333\\
};
\addlegendentry{\small $\hat{\mathbf{S}}$}

\addplot[ybar interval, fill=blue, fill opacity=0.6, draw=black, area legend] table[row sep=crcr] {%
x	y\\
0	0.00111111111111111\\
0.1	0.0822222222222222\\
0.2	0.0911111111111111\\
0.3	0.0822222222222222\\
0.4	0.0733333333333333\\
0.5	0.0666666666666667\\
0.6	0.0622222222222222\\
0.7	0.0555555555555556\\
0.8	0.0511111111111111\\
0.9	0.0477777777777778\\
1	0.0422222222222222\\
1.1	0.0422222222222222\\
1.2	0.0377777777777778\\
1.3	0.0333333333333333\\
1.4	0.0322222222222222\\
1.5	0.0311111111111111\\
1.6	0.0277777777777778\\
1.7	0.0255555555555556\\
1.8	0.0244444444444444\\
1.9	0.0188888888888889\\
2	0.0188888888888889\\
2.1	0.0155555555555556\\
2.2	0.0144444444444444\\
2.3	0.01\\
2.4	0.00888888888888889\\
2.5	0.00222222222222222\\
2.6	0\\
2.7	0\\
2.8	0\\
2.9	0\\
3	0\\
3.1	0.00111111111111111\\
3.2	0.00111111111111111\\
};
\addlegendentry{\small $\check{\mathbf{\Sigma}}_{FP}$}

\addplot [color=red, line width=0.5mm]
  table[row sep=crcr]{%
3.06108446691931	0\\
3.06108446691931	2\\
};
\addlegendentry{\small Threshold $t$}

\end{axis}
\end{tikzpicture}%
	\caption{Eigenvalues of the covariance matrices $\hat{\mathbf{S}}$ and $\check{\mathbf{\Sigma}}_{FP}$ for a single source with $SNR = 10dB$ present in the samples and the calculated threshold $\rho = 0.7$, $m=900$, $N=2000$, $\tau \sim $ inverse gamma}
	\label{fig9_bis}
\end{figure}
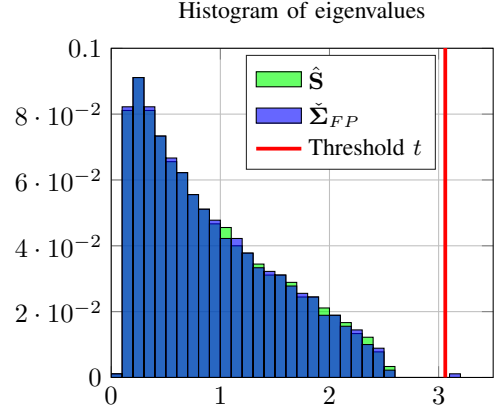

The results are better for this robust method than in the previous section (e.g. figure \ref{fig4}). Indeed, the robust method provides robustness with respect to the distribution of $\tau$: if the distribution of the texture differs to those for which the function $u$ has been computed, the method is still reliable, this can be explained by the robustness of the covariance matrix estimation. As for the non-whitening case, the eigenvalues get over the threshold and no conclusion or model order can be deducted.  These results have so  extended the paper of \cite{Vinogradova14} to the left hand side correlated noise case. The L2-norm of the estimated covariance matrix compared to the SCM tends to zero when $N$ and $m$  tends to infinity with a constant ratio $c$. As a lot of estimation methods for the rank of the signal subspace are based on the estimation of the eigenvalues of the covariance matrix, this new estimator improves the consistency for resolution of this problem.


\section{ Results and Comparisons}

In this section some results of order selection are presented, on both simulated and real hyperspectral images. The simulations are based on $\check{\mathbf{\Sigma}}_{SCM}$  and $\check{\mathbf{\Sigma}}_{FP}$.

\subsection{Estimation of the model order} 

In order to test the proposed method, we simulate hyperspectral images, before dealing with real images. 
As a reminder, we first whiten the received signal thanks to a Toeplitz matrix coming from the SCM or a Fixed-Point estimator. Thus, a $M$-estimator is used to estimate the scatter matrix of the whitened signal. The distribution of its eigenvalues is then studied: a threshold is applied to count how many eigenvalues are higher than this threshold, providing the estimated model order $\hat{p}$. \\

For simulated and correlated ($\rho=0.7$) CES noise,  the $\left\{\tau_i\right\}_{i\in \llbracket 0, N-1 \rrbracket}$ are  inverse gamma distributed with parameter $\nu=0.1$. On Figure \ref{fig12} ($m=400$ and $N=2000$), $p=4$ sources are added in the observations with a SNR varying from $-15$ to $20 $dB. For this figure, the number of sources $\hat{p}$ (average on 4 trials) is estimated  through three methods: AIC, the non-whitened signal and the two proposed methods: when the signal is whitened with the Toeplitz version of the SCM and the one of the FP. The proposed method starts to find sources from a SNR equal to $-5$dB. The FP method seems to better evaluate the number of sources. For a greater SNR, whereas it systematically gives the correct number of sources, the other methods overestimate it. On Figure \ref{fig13}  the same simulation is done for $p=4$ but with the $\left\{\tau_i\right\}_i$ following a distribution equal to $t^2$ where $t$ is a Student-t random variable, as before. On Figure \ref{fig13}, one notice that the proposed estimators still present better performance than the others, and allow to find sources with SNR greater than 0 dB.

\begin{figure}
	\centering 
%
%
\definecolor{mycolor1}{rgb}{1.00000,0.00000,1.00000}%
\definecolor{mycolor2}{rgb}{0.00000,1.00000,1.00000}%
\begin{tikzpicture}

\begin{axis}[%
width = 0.28\textwidth,
at={(0.758in,0.481in)},
scale only axis,
ticklabel style = {font=\tiny},
xmin=-15,
xmax=20,
xlabel style={font=\color{white!15!black}},
xlabel={\small SNR},
ymode=log,
ymin=1,
ymax=1000,
yminorticks=true,
ylabel style={font=\color{white!15!black}},
ylabel={\small number of sources},
axis background/.style={fill=white},
xmajorgrids,
ymajorgrids,
yminorgrids,
legend style= {at={(0.01,0.40)},anchor=west, legend cell align=left}
]
\addplot [color=blue, dashed, mark=asterisk, mark options={solid, blue}]
  table[row sep=crcr]{%
-15	55.5\\
-10	57.5\\
-5	57\\
0	59.5\\
3	60.5\\
5	62\\
7	63.5\\
10	64\\
20	65\\
};
\addlegendentry{\tiny{$\hat{p}$ with $\check{\mathbf{\Sigma}}_{UW}$}}

\addplot [color=red, dashed, mark=asterisk, mark options={solid, red}]
  table[row sep=crcr]{%
-15	0\\
-10	0\\
-5	1\\
0	3\\
3	4\\
5	4\\
7	4\\
10	4\\
20	4\\
};
\addlegendentry{\tiny{$\hat{p}$ with $\check{\mathbf{\Sigma}}_{FP}$}}

\addplot [color=mycolor1, dashed, mark=asterisk, mark options={solid, mycolor1}]
  table[row sep=crcr]{%
-15	0\\
-10	0\\
-5	1.5\\
0	4\\
3	4.5\\
5	4\\
7	4\\
10	5\\
20	4\\
};
\addlegendentry{\tiny{$\hat{p}$ with $\check{\mathbf{\Sigma}}_{SCM}$}}

\addplot [color=mycolor2, dashed, mark=asterisk, mark options={solid, mycolor2}]
  table[row sep=crcr]{%
-15	310\\
-10	314\\
-5	315\\
0	308.5\\
3	315\\
5	315.5\\
7	313.5\\
10	315\\
20	312.5\\
};
\addlegendentry{\tiny{$\hat{p}$ with AIC method}}

\addplot [color=green, dashed, mark=o, mark options={solid, green}]
  table[row sep=crcr]{%
-15	4\\
-10	4\\
-5	4\\
0	4\\
3	4\\
5	4\\
7	4\\
10	4\\
20	4\\
};
\addlegendentry{\tiny{p}}

\end{axis}
\end{tikzpicture}%
	\caption{ Estimation of the number $\hat{p}$ of sources (4 trials) embedded in CES correlated noise ($m=400$, $c=0.2$, $p=4$ source, $\rho=0.7$) versus SNR.}
	\label{fig12}
	\end{figure}
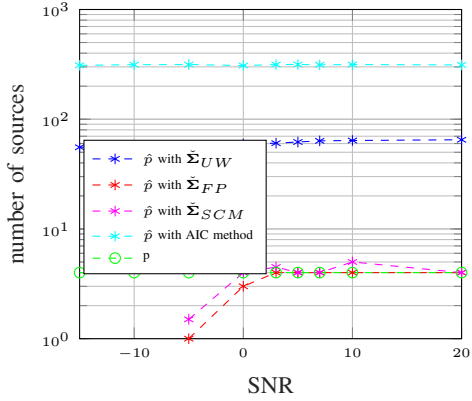

\begin{figure}
	\centering 
%
%
\definecolor{mycolor1}{rgb}{1.00000,0.00000,1.00000}%
\definecolor{mycolor2}{rgb}{0.00000,1.00000,1.00000}%
\begin{tikzpicture}

\begin{axis}[%
width = 0.28\textwidth,
at={(0.758in,0.481in)},
scale only axis,
ticklabel style = {font=\tiny},
xmin=-10,
xmax=10,
xlabel style={font=\color{white!15!black}},
xlabel={\small SNR},
ymode=log,
ymin=1,
ymax=1000,
yminorticks=true,
ylabel style={font=\color{white!15!black}},
ylabel={\small number of sources},
axis background/.style={fill=white},
xmajorgrids,
ymajorgrids,
yminorgrids,
legend style={at={(0.01,0.40)},anchor=west, legend cell align=left, align=left, draw=white!15!black}
]
\addplot [color=blue, dashed, mark=asterisk, mark options={solid, blue}]
  table[row sep=crcr]{%
-10	95\\
-5	94\\
0	96\\
1	97\\
2	97\\
3	97\\
5	96\\
7	97\\
10	100\\
};
\addlegendentry{\tiny{$\hat{p}$ with $\check{\mathbf{\Sigma}}_{UW}$}}

\addplot [color=red, dashed, mark=asterisk, mark options={solid, red}]
  table[row sep=crcr]{%
-10	0\\
-5	0\\
0	2\\
1	2\\
2	3\\
3	3\\
5	4\\
7	4\\
10	4\\
};
\addlegendentry{\tiny{$\hat{p}$ with $\check{\mathbf{\Sigma}}_{FP}$}}

\addplot [color=mycolor1, dashed, mark=asterisk, mark options={solid, mycolor1}]
  table[row sep=crcr]{%
-10	0\\
-5	0\\
0	2\\
1	3\\
2	3\\
3	3\\
5	4\\
7	4\\
10	4\\
};
\addlegendentry{\tiny{$\hat{p}$ with $\check{\mathbf{\Sigma}}_{SCM}$}}

\addplot [color=mycolor2, dashed, mark=asterisk, mark options={solid, mycolor2}]
  table[row sep=crcr]{%
-10	792\\
-5	778\\
0	789\\
1	773\\
2	781\\
3	797\\
5	790\\
7	789\\
10	785\\
};
\addlegendentry{\tiny{$\hat{p}$ with AIC method}}

\addplot [color=green, dashed, mark=o, mark options={solid, green}]
  table[row sep=crcr]{%
-10	4\\
-5	4\\
0	4\\
1	4\\
2	4\\
3	4\\
5	4\\
7	4\\
10	4\\
};
\addlegendentry{\tiny{p}}

\end{axis}
\end{tikzpicture}%
	\caption{Estimation of the number $p$ of sources (4 trials) embedded in CES correlated noise ($m=400$, $c=0.2$, $p=4$ sources, $\rho=0.7$) versus SNR.}
	\label{fig13}
\end{figure}
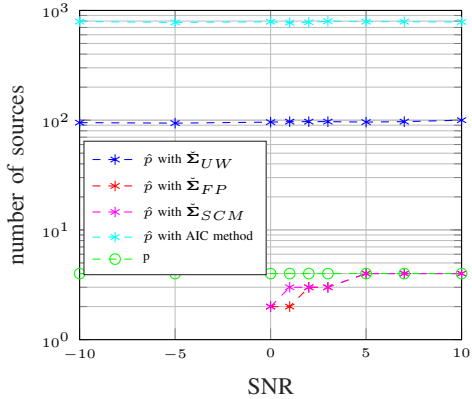

Now, we compare the results obtained with three different methods on several real hyperspectral images found in public access: {\it Indian Pines}, {\it SalinasA} from AVIRIS database and {\it PaviaU} from ROSIS database \cite{siteinternet}. Let $M1$ be the proposed method with a whitening made with the SCM estimator, $M2$ be the proposed method with a whitening made with a Fixed-Point estimator, $M3$ be the method consisting in thresholding the eigenvalues of the Fixed-Point estimator without the whitening step, and the usual AIC method. 
For the function $u(.)$ corresponding to the Student-t distribution, we choose  $\nu=0.1$ for the whitening process if it is done by a fixed-point estimator, and zero for the estimation process. As we do not have any access to the true distribution of the noise, an empirical estimator of $\gamma$ is used, $\hat{\gamma} = \displaystyle \frac{1}{N} \, \sum_{i=1}^{N} \frac{1}{m} \, \mathbf{y}_i^{H} \, \check{\mathbf{\Sigma}}^{-1}_{(i)} \, \mathbf{y}_i$, 
where $\check{\mathbf{\Sigma}}_{(i)} = \check{\mathbf{\Sigma}} - \displaystyle\frac{1}{N} \, u\left(\frac{1}{m} \,\mathbf{y}_i^{H} \, \check{\mathbf{\Sigma}}^{-1} \, \mathbf{y}_i\right) \, \mathbf{y}_i \, \mathbf{y}_i^H$. Then \cite{Couillet15b} shows that $\gamma -\hat{\gamma} \overset{a.s.}{\longrightarrow} 0$. Moreover, as the distribution of  $\tau$ is unknown, we choose to consider that $\mathds{E} \left[\tau\right] $ and $\mathds{E}\left[v(\tau \, \gamma) \, \tau \right]$ are equal to 1. Further works can be carried out to estimate correctly these unknown quantities. However, we can reasonably  assume than $\mathds{E}\left[v(\tau \, \gamma) \, \tau \right]$ and $\mathds{E} \left[\tau\right] $ are not to large and that the estimation error will not impact the results a lot.  The results are summarized in table \ref{table-Tableau1}.  On each image, the result tends to be better than those of classical methods.    

\begin{center}
\begin{table}[H]
\caption {Estimated $p$ for different hyperspectral images.}
\begin{center}
\renewcommand{\arraystretch}{0.7}
{\setlength{\tabcolsep}{0.25cm} 
\begin{tabular}{c|cccc}
\hline
Images & Indian Pines &  SalinasA &  PaviaU &  Cars \\
\hline
 $p$ & 16 & 9 & 9 & 6  \\
 $\hat{p}$ M1 & 11 & 9 & 1 & 3 \\
 $\hat{p}$ M2 & 12 & 9 & 1 & 3 \\
 $\hat{p}$ M3 &220 & 204 & 103 & 1 \\
 $\hat{p}$ AIC& 219 & 203 & 102& 143 \\
\hline
\end{tabular}}
\end{center}
\label{table-Tableau1}
\end{table}
\end{center}

\section{Conclusion}

The model order selection for large dimensional data and for sources embedded in correlated CES noise is tackled in this article. Two Toeplitz-based covariance matrix estimators are 
first introduced, and their consistency has been proved. As for the CES texture, it is handled with any $M$-estimator, which can then be used to estimate the correct structure of the scatter matrix 
built on whitened observations. The Random Matrix Theory provides tools to correctly estimate the model order. Results obtained on real and simulated hyperspectral images are 
promising. Moreover, the proposed method can be applied on a lot of other kind of model order selection problems such as radar clutter rank estimation, sources localization or any 
hyperspectral problems such as anomaly detection or  linear or non-linear unmixing techniques.

\begin{appendices}

\section{Proofs of Theorem 1 and Theorem 3}
\label{sec::4}
The proofs of Theorem 1 and Theorem 3 are inspired by \cite{Vinogradova14}. 
For these theorems, we will use the lemma 4.1 in \cite{Gray6}, that is, for $\mathbf{T} = \mathcal{L}\left(\left(t_0, \ldots, t_{m-1}\right)^T\right)$ a Toeplitz Hermitian $m \times m$-matrix with $\left\{t_{k}\right\}_{k\in \llbracket0,m-1\rrbracket}$ absolutely summable ($t_{-k} = t_k^\ast$), we can define the function $f(.)$ such that for any $\lambda \in[0,2\,\pi]$, $f(\lambda) = \displaystyle \sum_{k=1-m}^{m-1} t_k \, e^{i \lambda k} $ and $M_f$ characterizes its essential supremum:
\begin{equation}
\Vert \mathbf{T} \Vert \leq M_f = \underset{\lambda \in \left[  0,2 \pi \right)  }{\sup} \left\vert \displaystyle \sum_{k=1-m}^{m-1} t_k \, e^{i \lambda k} \right\vert \, .
\label{eqGray}
\end{equation}

\subsection{Proof of Theorem 1}
As in the main body of this article, let $\mathbf{Y} = \mathbf{M} \, \boldsymbol{\delta}^H \, \mathbf{\Gamma}^{1/2}  + \mathbf{C}^{1/2}\,\mathbf{ X}\,\mathbf{ T}^{1/2}$ and let $\mathcal{T}$ be the Toeplitz operator as defined in the introduction and for any $m$-vector $\mathbf{x}$, $\left( \left[\mathcal{L} (\mathbf{x})\right]_{i,j} \right) _{i\leq j} = x_{i-j}$ and $\left( \left[\mathcal{L} (\mathbf{x})\right]_{i,j} \right)_{i>j} = x_{i-j}^\ast$, of size $m \times m$. Under Assumption 1,  Assumption 2, Assumption 3 and as $\mathcal{T} \left(\displaystyle \frac{1}{N} \, \mathbf{Y}\,  \mathbf{Y}^H\right) $ and $\mathds{E}[\tau] \, \mathbf{C}$ are Toepltiz matrices, one can write, thanks to \eqref{eqGray}: 
\begin{equation}
\left\Vert \mathcal{T} \left(\frac{1}{N} \, \mathbf{Y}\,  \mathbf{Y}^H\right) -  \mathds{E}[\tau]\,  \mathbf{C} \right\Vert \leq \underset{\lambda \in \left[  0,2 \pi \right)  }{\sup} \left\vert  \hat{\gamma}_m (\lambda) - \mathds{E}[\tau] \,\gamma_m (\lambda)   \right\vert \, ,
\label{eqproof1}
\end{equation}
where $\mathbf{\gamma}_m (\lambda) =  \displaystyle \sum_{k =1-m}^{m-1} c_{k,m} \, e^{i\,k \,\lambda}$ with $c_{-k} = c_k^{\star}$ and $\hat{\mathbf{\gamma}}_m (\lambda) =  \displaystyle \sum_{k =1-m}^{m-1} \check{c}_{k,m} \, e^{i\,k \,\lambda}$ with $\check{c}_{-k} = \check{c}_k^{\star}$. \\

The following lemma is essential for the development of the proof:
 \begin{lemma}
 \label{lemma1}
 The quantity $\hat{\mathbf{\gamma}}_m (\lambda)$ can be rewritten as:
\begin{equation}
\hat{\mathbf{\gamma}}_m (\lambda) = \mathbf{d}_m^H (\lambda) \, \frac{\mathbf{Y}\,  \mathbf{Y}^H}{N} \, \mathbf{d}_ m (\lambda) \, ,
\label{eq1th1lemme1}
\end{equation} 
with   $\mathbf{d}_m (\lambda) =\displaystyle \frac{1}{\sqrt{m}} \, \left(1, e^{-i \,\lambda}, \ldots, e^{-i\,(m-1) \,\lambda} \right)^T$. 
\end{lemma}
 
\begin{proof}
The proof draws his inspiration from the one of Appendix A1 in \cite{Vinogradova14}. Equation \eqref{eq1th1lemme1} can be rewritten as:
\begin{eqnarray}
 & \mathbf{d}_m^H (\lambda) & \,\displaystyle\frac{\mathbf{Y} \,\mathbf{Y}^H}{N} \,\mathbf{d}_m (\lambda)  =  \frac{1}{m\,N} \displaystyle \sum_{l,l' =0}^{m-1} e^{-i\,(l'-l)\,\lambda} \, \left[\mathbf{Y} \,\mathbf{Y}^H \right]_{l,l'} \, ,\nonumber \\
&= &\displaystyle \sum_{k= 1-m}^{m-1} e^{-i\,k \,\lambda} \,\frac{1}{m \,N} \displaystyle \sum_{i=0}^{m-1} \displaystyle \sum_{j=0}^{N-1} y_{i,j}  \, y_{i+k,j}^\star \mathds{1}_{0 \leq i+k\leq  m } \, ,\nonumber \\
& = & \displaystyle \sum_{k= 1-m}^{m-1} \check{c}_k \, e^{-i\,k \,\lambda} \nonumber \, .
\end{eqnarray}
\end{proof}

\vspace{-1em}
Thereby, we have:
\vspace{-1em}

\begin{eqnarray}
& \hat{\mathbf{\gamma}}&_m (\lambda)  =  \mathbf{d}_m^H (\lambda) \, \frac{\mathbf{Y}\,  \mathbf{Y}^H}{N} \, \mathbf{d}_m (\lambda)\, , \nonumber  \\ 
 &=&  \mathbf{d}_m^H (\lambda)\, \frac{\mathbf{M} \boldsymbol{\delta}^H  \, \mathbf{\Gamma} \, \boldsymbol{\delta} \, \mathbf{M}^H}{N} \, \mathbf{d}_m (\lambda) \nonumber \\ 
 &+&  \mathbf{d}_m^H (\lambda) \, \frac{\mathbf{C}^{1/2}\,  \mathbf{X} \, \mathbf{T}^{1/2}  \, \mathbf{\Gamma}^{1/2} \, \boldsymbol{\delta}  \, \mathbf{M}^H}{N}  \, \mathbf{d}_m (\lambda) \nonumber \\ 
  &+&  \mathbf{d}_m^H (\lambda) \, \frac{\mathbf{M}  \, \boldsymbol{\delta}^H  \mathbf{\Gamma}^{1/2} \, \mathbf{T}^{1/2} \, \mathbf{X}^H \, \mathbf{C}^{1/2}}{N}  \, \mathbf{d}_m (\lambda)   \nonumber \\ 
 &+&  \mathbf{d}_m^H (\lambda) \, \frac{\mathbf{C}^{1/2}  \, \mathbf{X} \,  \mathbf{T} \, \mathbf{X}^H   \, \mathbf{C}^{1/2}}{N} \, \mathbf{d}_m (\lambda) \, \label{eqcomp} . 
\end{eqnarray}

And we note : $ \hat{\mathbf{\gamma}}_m^{sign} (\lambda) = \mathbf{d}_m^H (\lambda)\, \frac{\mathbf{M} \boldsymbol{\delta}^H  \, \mathbf{\Gamma} \, \boldsymbol{\delta} \, \mathbf{M}^H}{N} \, \mathbf{d}_m (\lambda)$, \\

$\hat{\mathbf{\gamma}}_m^{cross} (\lambda) =$ \\ 
\indent $\mathbf{d}_m^H (\lambda) \, \frac{\mathbf{C}^{1/2}\,  \mathbf{X} \, \mathbf{T}^{1/2}  \, \mathbf{\Gamma}^{1/2} \, \boldsymbol{\delta}  \, \mathbf{M}^H + \mathbf{M}  \, \boldsymbol{\delta}^H  \mathbf{\Gamma}^{1/2} \, \mathbf{T}^{1/2} \, \mathbf{X}^H \, \mathbf{C}^{1/2}}{N}  \, \mathbf{d}_m (\lambda)$   \\

$\hat{\mathbf{\gamma}}_m^{noise} (\lambda) =  \mathbf{d}_m^H (\lambda) \, \frac{\mathbf{C}^{1/2}  \, \mathbf{X} \,  \mathbf{T} \, \mathbf{X}^H   \, \mathbf{C}^{1/2}}{N} \, \mathbf{d}_m (\lambda) $ .

And the equation \eqref{eqproof1} becomes: 
\begin{eqnarray}
 && \left\Vert  \mathcal{T} \left(\frac{1}{N} \, \mathbf{Y}\,  \mathbf{Y}^H\right) -  \mathds{E}[\tau]\,  \mathbf{C} \right\Vert \leq  \underset{\lambda \in \left[  0,2 \pi \right)  }{\sup} \vert  \hat{\gamma}_m^{noise} (\lambda) \nonumber \\
 && +  \hat{\gamma}_m^{sign} (\lambda) + \hat{\gamma}_m^{cross} (\lambda) - \mathds{E}[\tau] \,\gamma_m (\lambda)   \vert \, . 
\label{eqproof3}
\end{eqnarray}

This leads to: 
\begin{eqnarray}
&& \left\Vert \mathcal{T} \left(\frac{1}{N} \mathbf{Y} \, \mathbf{Y}^H\right) -  \mathds{E}[\tau]\,  \mathbf{C} \right\Vert \nonumber \\
&& \leq  \underset{\lambda \in \left[  0,2 \pi \right)  }{\sup} \left\vert  \hat{\gamma}_m^{noise} (\lambda)  - \mathds{E} \left[\hat{\gamma}_m^{noise} (\lambda) \right] \right\vert \nonumber  \\
&&  + \underset{\lambda \in \left[  0,2 \pi \right)  }{\sup} \left\vert  \mathds{E} \left[\hat{\gamma}_m^{noise} (\lambda)\right] - \mathds{E}(\tau) \, \gamma_m (\lambda)  \right\vert  \nonumber \\ 
&& +  \underset{\lambda \in \left[  0,2 \pi \right)  }{\sup} \left\vert \hat{\gamma}_m^{sign} (\lambda)\right\vert + \underset{\lambda \in \left[  0,2 \pi \right)  }{\sup} \left\vert \hat{\gamma}_m^{cross} (\lambda) \right\vert \, .
\label{eqproof5}
\end{eqnarray}

We will now analyse each term of \eqref{eqproof5}. \\ 

 \subsubsection{Analysis of $\underset{\lambda \in \left[  0,2 \pi \right)  }{\sup} \left\vert  \mathds{E} \left[\hat{\gamma}_m^{noise} (\lambda)\right] - \mathds{E}(\tau) \, \gamma_m (\lambda)  \right\vert$ }

We first need the following lemma:

\begin{lemma}
\label{lemma2}
\begin{equation}
\mathds{E} \left[\hat{\gamma}_m^{noise} (\lambda) \right] = \mathds{E}[\tau] \, \mathbf{d}_m^H (\lambda) \, \mathbf{C} \, \mathbf{d}_m (\lambda) = \mathds{E}[\tau] \, \gamma_m (\lambda)\, .
\label{eq2th1lemme1}
\end{equation}
\end{lemma}

\begin{proof}
The equation \eqref{eqcomp} gives
$\mathds{E} \left[\hat{\gamma}_m^{noise} (\lambda)\right] = \mathbf{d}_m^H(\lambda) \, \mathds{E}  \left[  \displaystyle \frac{\mathbf{C}^{1/2} \mathbf{X} \, \mathbf{T} \, \mathbf{X}^H \mathbf{C}^{1/2}}{N} \right]\,  \mathbf{d}_m (\lambda)$. Let $\mathbf{V} = \mathbf{C}^{1/2} \mathbf{X}$ and $\left(\mathbf{V}\right)_{i,j} = v_{i,j}$. We obtain 
$\mathds{E} \left[\hat{\gamma}_m^{noise} (\lambda)\right] = \mathbf{d}_m^H(\lambda) \, \mathds{E}  \left[  \displaystyle \frac{ \mathbf{V} \, \mathbf{T} \, \mathbf{V}^H }{N} \right]\,  \mathbf{d}_m (\lambda)$. As $ \left( \mathds{E} \left[\mathbf{V} \,\mathbf{T} \,\mathbf{V}^H \right]\right)_{i j} =  \displaystyle \sum_{k=1}^N  \mathds{E}[\tau]\,   \mathds{E} \left[v_{j,k}^{\ast} \, v_{ik}\right]$ and  $c_{k'} = \mathds{E} \left[v_{p, n} \,v_{p+k',n}^{\ast}\right]$, we have $\left( \mathds{E} \left[\mathbf{V} \, \mathbf{T} \, \mathbf{V}^H\right] \right)_{i j}  =  \displaystyle \sum_{k=1}^N \mathds{E}[\tau]\,  c_{j-i} = N \,\mathds{E}[\tau] \, c_{j-i}$. This leads to
\begin{equation*}
\mathbf{\mathds{E}} \left[\hat{\gamma}_m^{noise} (\lambda)\right] = \frac{N\, \mathds{E}(\tau) }{N}  \,\mathbf{d}_m^H (\lambda)  \, \mathbf{C}  \, \mathbf{d}_m (\lambda) = \mathds{E}[\tau] \,  \gamma_m (\lambda) \, .
\end{equation*}
\end{proof}


Thereby, the second term of \eqref{eqproof5} leads to: 
\begin{align*}
& \underset{\lambda \in \left[  0,2 \pi \right)  }{\sup} \left\vert  \mathds{E} \left[ \hat{\gamma}_m^{noise} (\lambda) \right] - \mathds{E}(\tau)  \,\gamma_m(\lambda)  \right \vert \\
& =  \underset{\lambda \in \left[  0,2 \pi \right)  }{\sup} \mid  \mathds{E}(\tau)   \,\gamma_m(\lambda) -  \mathds{E}(\tau) \, \gamma_m (\lambda)  \mid  = 0\, .
\end{align*}
The second term is equal to zero.


 \subsubsection{Analysis of $\underset{\lambda \in \left[  0,2 \pi \right)  }{\sup} \left\vert  \hat{\gamma}_m^{noise} (\lambda)  - \mathds{E} \left[\hat{\gamma}_m^{noise} (\lambda) \right] \right\vert $}
 
As in \cite{Vinogradova14}, the method consists in proving for a $\lambda_i \in \left[  0,2 \pi \right) $ and a real $ x>0 $ that $ \mathds{P} \left[ \left\vert  \hat{\gamma}_m^{noise} (\lambda_i)  - \mathds{E} \left[ \gamma_m^{noise} (\lambda_i)  \right] \right\vert > x \right] \rightarrow 0$. After that, it remains to prove that $ \mathds{P} \left[  \underset{\lambda \in \left[  \lambda_i , \lambda_{i+1} \right)  }{\sup} \left\vert \gamma_m^{noise} (\lambda) - \gamma_m^{noise} (\lambda_i)  \right\vert > x \right] \rightarrow 0 $ and that  $ \mathds{P} \left[  \underset{\lambda \in \left[  \lambda_i , \lambda_{i+1} \right)  }{\sup} \left\vert \mathds{E} \gamma_m^{noise} (\lambda_i) - \mathds{E} \gamma_m^{noise} (\lambda)  \right\vert > x \right] \rightarrow 0$.
 \noindent Let $\lfloor\cdot\rfloor $ be the floor function, choosing a $\beta > 2 $,  $\mathcal{I} = \left[  0, \ldots, \lfloor N^{\beta} \rfloor -1 \right] $,  $\lambda_i = 2\, \pi \,\frac{i}{\lfloor N^{\beta}\rfloor }$, $i \in \mathcal{I}$: 
\begin{eqnarray}
&& \underset{\lambda \in \left[  0,2 \pi \right)  }{\sup} \left\vert  \hat{\gamma}_m^{noise} (\lambda) -  \mathds{E} \left[\hat{\gamma}_m^{noise} (\lambda)\right] \right\vert \nonumber \\
 && \leqslant  \max_{i \in \mathcal{I}} \underset{\lambda \in \left[  \lambda_i \lambda_{i+1} \right] }{\sup} \left\vert  \hat{\gamma}_m^{noise} (\lambda) -  \hat{\gamma}_m^{noise} (\lambda_i) \right\vert  \nonumber \\
& & + \max_{i \in \mathcal{I}} \left\vert  \hat{\gamma}_m^{noise} (\lambda_i) -  \mathds{E} \left[\hat{\gamma}_m^{noise} (\lambda_i)\right] \right\vert \nonumber \\
& & + \max_{i \in \mathcal{I}}  \underset{\lambda \in \left[  \lambda_i \lambda_{i+1} \right] }{\sup} \left\vert \mathds{E} \left[ \hat{\gamma}_m^{noise} (\lambda_i)\right] -  \mathds{E} \left[ \hat{\gamma}_m^{noise} (\lambda)\right] \right\vert  \nonumber \, ,\\
 & \overset{\triangle}{=} & \chi_1 + \chi_2 +\chi_3\, .
 \label{eqgamma}
\end{eqnarray}

\noindent The idea of the proof in \cite{Vinogradova14} is then to provide concentration inequalities for the term $\chi_1$ and $\chi_2$ (random terms) and a bound on $\chi_3$. The only difference with \cite{Vinogradova14} is the presence of the matrix $\mathbf{T}$  in $\hat{\gamma}_m^{noise}(\lambda)$ and the left side correlation of the noise. Let note $ \Vert \boldsymbol{\gamma} \Vert_{\infty}$ the sup norm of the function $\boldsymbol{\gamma}: \lambda \longrightarrow \displaystyle \sum_{k=-\infty}^{\infty}c_k \, e^{-ik\lambda}$ for $\lambda \in [0 \, 2\pi )$. The convergence of the first term $\chi_1$ is proposed in the following lemma.

\begin{lemma}
\label{lemma3}
A constant $A > 0$ can be found
\end{lemma}
such that, for any $x>0$ and $N$ large enough, 
\footnotesize
\begin{equation*}
\mathds{P} \left[  \chi_1 > x \right]  \leq \exp \left( - \frac{c\,N^2}{\left\Vert\mathbf{T} \right\Vert_{\infty}} \, \left( \frac{x\,N^{\beta-2}}{A \, \left\Vert \boldsymbol{\gamma} \right\Vert_{\infty}} - \log\left( \frac{x\,N^{\beta-2}}{A \, \left\Vert \boldsymbol{\gamma} \right\Vert_{\infty}}\right) -1  \right)  \right)   \, .
\end{equation*}
\normalsize

\begin{proof}
As already mentioned, the proof is the same as in \cite{Vinogradova14} except for two points: the presence of the matrix $\mathbf{T}$ and the left side correlation of the noise instead of right side in \cite{Vinogradova14}. The inequality: $\left\Vert \mathbf{V}_N \,\mathbf{T} \,\mathbf{V}^H_N \right\Vert \leqslant  \left\Vert \mathbf{T} \right\Vert_{\infty} \, \left\Vert \mathbf{V}_N \,\mathbf{V}^H_N \right\Vert  \leq \left\Vert \mathbf{T} \right\Vert_{\infty}  \, \left\Vert \mathbf{C} \right\Vert \, \left\Vert \mathbf{X\, X}^H \right\Vert$ enables to write:
\begin{align*}
& \left\vert \hat{\gamma}_m^{noise} (\lambda) -  \hat{\gamma}_m^{noise} (\lambda_i) \right\vert \\
&= \left\vert \mathbf{d}_m^H(\lambda) \, \frac{\mathbf{V} \, \mathbf{T} \, \mathbf{V}^H}{N} \, \mathbf{d}_m(\lambda) -  \mathbf{d}_m^H(\lambda_i) \, \frac{\mathbf{V} \, \mathbf{T} \, \mathbf{V}^H}{N}  \, \mathbf{d}_m(\lambda_i)  \right\vert\, , \\
&\leq \frac{2}{N} \, \left\vert  \mathbf{d}_m(\lambda) - \mathbf{d}_m(\lambda_i) \right\vert \, \left\Vert \mathbf{C} \right\Vert \, \left\Vert \mathbf{T} \right\Vert_\infty  \, \left\Vert \mathbf{X\, X}^H  \right\Vert \, .
\end{align*}
And then the end of the proof is exactly the same as those of the Lemma 4 in \cite{Vinogradova14}  replacing $c$ by $\displaystyle\frac{c}{\left\Vert \mathbf{T} \right\Vert_\infty}$ in the exponential. The left correlation is without consequences on the proof. 
\end{proof}
The convergence of the second term $\chi_2$ is proposed in the following lemma.

\begin{lemma}
\label{lemma4}
\footnotesize
\begin{equation*}
\mathds{P} \left[  \chi_2 > x \right]  \leq 2\,N^{\beta}\, \exp \left( - \frac{c\,N}{\left\Vert \mathbf{T} \right\Vert_{\infty} } \, \left( \frac{x}{\left\Vert  \boldsymbol{\gamma} \right\Vert_{\infty}} - \log\left(\frac{x}{ \left\Vert  \boldsymbol{\gamma} \right\Vert_{\infty}} + 1 \right) \right)  \right)  \, .
\end{equation*}
\end{lemma}

\normalsize

\begin{proof}
The proof is the same as those of the Lemma 5 in \cite{Vinogradova14}, with the $\displaystyle\frac{c}{\left\Vert \mathbf{T} \right\Vert_\infty}$  on the denominator.
\end{proof}

The convergence of the third term $\chi_3$ is proposed in the following lemma.

\begin{lemma}
\label{lemma5}
\begin{equation*}
\chi_3 \leq A \left\Vert \boldsymbol{\gamma} \right\Vert_{\infty} \, N^{- \beta + 1}\, .
\end{equation*}
\end{lemma}

\begin{proof}
The proof is the same as those of the lemma 6 of  \cite{Vinogradova14}, still with the $\displaystyle\frac{c}{\Vert \mathbf{T} \Vert_\infty}$  on the denominator.
\end{proof}

These inequalities proves that $\mathds{P} \left[   \underset{\lambda \in \left[  0,2 \pi \right)  }{\sup} \left\vert  \hat{\gamma}_m^{noise} (\lambda)  - \mathds{E} \left[\hat{\gamma}_m^{noise} (\lambda) \right] \right\vert  > x \right] \overset{a.s.}{\longrightarrow} 0$ for any $x$ positive real  and with a $e^{-N^2}$ rate of decrease. \\


\subsubsection{Analysis of  $\underset{\lambda \in \left[  0,2 \pi \right)  }{\sup} \left\vert \hat{\gamma}_m^{cross} (\lambda) \right\vert$}

To prove the convergence of the last term of \eqref{eqproof5}, let us recall that
\begin{align*}
\hat{\gamma}_m^{cross} (\lambda) & = \mathbf{d}_m^H (\lambda) \, \frac{\mathbf{C}^{1/2}\,  \mathbf{X} \, \mathbf{T}^{1/2}  \, \mathbf{\Gamma}^{1/2} \, \boldsymbol{\delta}  \, \mathbf{M}^H}{N} \, \mathbf{d}_ m (\lambda)\\
& +  \mathbf{d}_m^H (\lambda) \frac{\mathbf{M}  \, \boldsymbol{\delta}^H  \mathbf{\Gamma}^{1/2} \, \mathbf{T}^{1/2} \, \mathbf{X}^H \, \mathbf{C}^{1/2}}{N}  \, \mathbf{d}_ m (\lambda) \, .
\end{align*}
\noindent Let $\mathds{I}_m$ be a $m \times m$ matrix containing 1 everywhere and $\mathbf{D}_m(\lambda)$ be the matrix containing the elements of $\mathbf{d}_m(\lambda)$ on its diagonal. It can be easily verified that, for any matrix $\mathbf{A}$, $\mathbf{d}_m^H(\lambda)\,  \mathbf{A} \, \mathbf{d}_m(\lambda) = \mathrm{Tr}\left(\mathbf{D}_m^H(\lambda) \, \mathbf{A} \, \mathbf{D}_m(\lambda) \, \mathds{I}_m\right)$. We obtain:
\begin{align*}
& \hat{\gamma}_m^{cross} (\lambda) = \\
 & 2 \, \mathcal{R}e \left[   \frac{1}{N}  \, \mathrm{Tr} \left(\mathbf{X} \, \mathbf{T}^{1/2} \,\boldsymbol{\Gamma}^{1/2} \, \boldsymbol{\delta} \, \mathbf{M}^H \, \mathbf{D}_m (\lambda)  \,  \mathds{I}_{m} \, \mathbf{D}_m^H (\lambda) \, \mathbf{C}^{1/2} \right) \right] \, .
\end{align*}
For readability, let $\mathbf{E}(\lambda) = \mathbf{D}_m(\lambda)   \, \mathds{I}_{m} \, \mathbf{D}_m^H(\lambda)$ defined as:

\[
\mathbf{E}(\lambda)=
  \begin{pmatrix}
   1 & e^{i \lambda}  & \ldots & e^{i(m-1)\lambda} \\
   e^{-i \lambda} & 1 & \ldots & e^{i(m-2) \lambda} \\
    e^{-i(m-1)\lambda} &  \ldots & \ldots &  1
  \end{pmatrix}  \, ,
\]

\noindent let $\mathbf{G}(\lambda) = \mathbf{M}^H \, \mathbf{E}(\lambda) \,\mathbf{C}^{1/2} \,\mathbf{X}$ and $\mathbf{J} =  \mathbf{T}^{1/2} \, \boldsymbol{\Gamma}^{1/2} \, \boldsymbol{\delta}$, two matrices respectively of size $p \times N$ and $N \times p$. Moreover, let $\mathbf{g}(\lambda) = \left[g_1(\lambda), \ldots, g_{N\, p}(\lambda) \right]^T = \mathrm{vec}(\mathbf{G}(\lambda))$  and $\mathbf{j}= \left[j_1, \ldots, j_{N\, p} \right]^T= \mathrm{vec}(\mathbf{J})$. We obtain:
\begin{align*}
\hat{\gamma}_m^{cross} (\lambda) 
&= \frac{2}{N} \, \mathcal{R}e \left(  \mathrm{vec}^T(\mathbf{G}(\lambda)) \, \mathrm{vec}(\mathbf{J})   \right)\, , \\
&= \frac{2}{N}  \, \displaystyle \sum_{k=1}^{N \, p} \mathcal{R}e (g_k(\lambda) )  \, \mathcal{R}e (j_k) - \mathcal{I}m (g_k(\lambda)) \, \mathcal{I}m (j_k)\, .
\end{align*}

This expression can be transformed by introducing $\mathbf{A} = \mathbf{M}^H  \, \mathbf{E}  \, \mathbf{C}^{1/2} \otimes \mathbf{I}_N$, $\mathbf{B} = \mathbf{T}^{1/2}  \, \boldsymbol{\Gamma}^{1/2} \otimes \mathbf{I}_p$, $\tilde{\mathbf{g}}(\lambda) =   \mathbf{A}^{-1}\, \mathbf{g}(\lambda)$ \,  $\tilde{\mathbf{j}} = \mathbf{B}^{-1}  \, \mathbf{j}$, $a_k =  \displaystyle \sum_{l=1}^{N \, p} \left( \mathbf{A}^{T}\right)_{l,k}$ and $b_k = \displaystyle \sum_{s=1}^{N \, p}  \left(\mathbf{B}\right)_{s,k}$:
\begin{align*}
&\hat{\gamma}_m^{cross} (\lambda)  \\
& =\frac{2}{N} \, \mathcal{R}e \left(   \left(\mathbf{A}^{-1}\, \mathbf{g}(\lambda)\right)^{T} \,  \mathbf{A}^{T} \,  \mathbf{B} \, \left(   \mathbf{B}^{-1} \, \mathbf{j}\right)   \right) \, \\
& = \displaystyle \frac{2}{N}  \,  \sum_{k=1}^{N \, p}   \mathcal{R}e \left( \displaystyle \sum_{l=1}^{N \,  p}  \left(\mathbf{A}^{T}\right)_{l,k} \, \tilde{\mathbf{g}}_k (\lambda)    \right) \, \mathcal{R}e \left( \displaystyle \sum_{s=1}^{N \, p}   \left(\mathbf{B}\right)_{s,k} \, \tilde{\mathbf{j}}_k  \right) \\
&- \mathcal{I}m \left(   \displaystyle \sum_{l=1}^{N \, p} \left(\mathbf{A}^{T}\right)_{l,k} \, \tilde{\mathbf{g}}_k (\lambda)   \right)  \,   \mathcal{I}m \left( \displaystyle \sum_{s=1}^{N \,  p}    \left( \mathbf{B}\right)_{s,k} \tilde{\mathbf{j}}_k   \right) \, , \\
& = \displaystyle \frac{2}{N} \,  \sum_{k=1}^{N \, p} \,      \mathcal{R}e \left( a_k \, \tilde{g}_k (\lambda)  \right) \, \mathcal{R}e \left( b_k \, \tilde{j}_k   \right) \\
&- \mathcal{I}m \left(  a_k \, \tilde{g}_k (\lambda)   \right)    \, \mathcal{I}m \left( b_k \, \tilde{j}_k   \right) \, . 
\end{align*}
The variables $a_k \, \tilde{g}_k (\lambda)$ and $b_k \, \tilde{j}_k$ are two independent complex Gaussian variables with variances respectively equal to $\left|\tilde{a}_k(\lambda)\right|^2$ and $\left|b_k\right|^2$.  We can apply the following lemma:

\begin{lemma}
\label{lemma6}
Let $x$ and $y$ be two independent Gaussian $\mathcal{N}(0,1)$
\end{lemma}
\vspace{-0.5em}
\noindent \textit{ scalar random variables, then for any $\tau \in (-1 \,\, 1)$, then
$\mathds{E}\left[\exp{\left(\tau \, x \, y\right)}\right]= (1- \tau^2)^{-1/2}$.}

\begin{proof}
The proof is derived in \cite{Vinogradova14} through lemma 13.
\end{proof}

\noindent Let $\nu >0$ a real such that : $\nu^{-1} >  \underset{k\in \llbracket1, N p\rrbracket, \lambda \in \left[  0,2 \pi \right)  }{\sup} \left(\left\vert \tilde{a}_k(\lambda) \right\vert^2 \,  \left\vert b_k \right\vert^2\right)$. Then, for a fixed $\lambda \in [0 \, \,2\pi  )$, from  Lemma \ref{lemma6} and from the Markov Inequality:
\begin{align*}
& \mathds{P} \left[   \hat{\gamma}_m^{cross} (\lambda) > x  \mid \mathbf{T} \right] \\ 
&= \mathds{P} \left[  \exp{ \left( N \,\nu \,\hat{\gamma}_m^{cross} (\lambda) \right)} > \exp{ \left(N \,\nu \,x\right) } \, \mid \mathbf{T} \right] \\
&  \leq \exp{\left(-N \,\nu \,x\right)} \, \mathds{E} \left[ \exp \left(2 \,\nu    \displaystyle \sum_{k=1}^{N \, p} \,  \left[     \mathcal{R}e \left( a_k \tilde{g}_k (\lambda)    \right) \, \mathcal{R}e \left( b_k \tilde{j}_k   \right) \right. \right. \right. \\
& \left. \left. \left. - \mathcal{I}m \left(  a_k \, \tilde{g}_k (\lambda)   \right)    \, \mathcal{I}m \left( b_k \, \tilde{j}_k   \right)     \right] \right) \right] \\
&\leq \exp{ \left(-N \, \nu \, x\right)} \displaystyle \prod_{k=1}^{Np} \left(  1 - 4 \, \nu^2   \, \frac{\vert \tilde{a}_k(\lambda) \vert^2}{2} \,   \frac{\vert b_k \vert^2}{2}  \right)^{-1/2}  \, \\
& \left(  1 - 4 \,\nu^2  \, \frac{\vert \tilde{a}_k(\lambda) \vert^2}{2}  \, \frac{\vert b_k \vert^2}{2}  \right)^{-1/2} \\
&\leq \exp{\left(-N \nu x - \displaystyle \sum_{k=1}^{N\, p} \, \log{\left( 1 - \nu^2  \, \vert \tilde{a}_k(\lambda) \vert^2 \,  \vert b_k \vert^2  \right)^{-1} }\right)}
\end{align*}

\noindent Moreover, since the $\Gamma_{i,j}$ are absolutely summable (Assumption 3), it exists  a constant $K$ such that:
\begin{equation*}
\left\vert  b_k \right\vert ^2 =  \left\vert  \displaystyle \sum_{l=1}^{N p} \sqrt{\tau_l} \,  \Gamma_{l,k}^{1/2} \right\vert^2 \leq K \displaystyle \sum_{l=1}^{N p}   \tau_l \, .  
\end{equation*}
Furthermore, since $\displaystyle\frac{1}{N} \, \displaystyle \sum_{l=1}^{N  p}   \tau_l  \underset{N \longrightarrow \infty}{\longrightarrow}  \mathds{E} (\tau_i) = 1$, \
we obtain $\vert  b_k \vert ^2 \leq N \, K$. To deal with $\vert \tilde{a}_k(\lambda) \vert$, let $K_1$ and $K_2$  be some constants and remind that, for a fixed $j$, the $\left\{c_{i,j}\right\}_i$ and the $\left\{M_{i,j}\right\}_i$ are absolutely summable:
\begin{align*}
\left\vert a_k(\lambda) \right\vert &= \frac{1}{m} \, \left\vert  \displaystyle \sum_{s=1}^p \sum_{l,j =1}^m c_{l,k} \, M_{j,s}^{\star} \, e^{i(l-j)\,\lambda} \right\vert \, ,\\
&\leq \frac{1}{m} \displaystyle \sum_{l=1}^m  \left\vert c_{l,k}\right\vert \, m \, \sum_{s,j =1}^{p,m} \left\vert M_{j,s}^{\star} \right\vert\, ,\\
&\leq p \, K_2 \, \underset{s}{\max}  \left(\displaystyle  \sum_{j =1}^m \left\vert M_{j,s}^{\star} \right\vert\right) =  p\, K_1  \, .
\end{align*}

\noindent We obtain $\nu ^2 \, \left\vert \tilde{a}_k(\lambda) \right\vert^2 \, \left\vert b_k\right\vert^2  \leq \nu^2 \, N^2 \, p^2 \, K \, K_1^2$ with $p \ll N $. Let $q$ and $\epsilon$ be two positive reals small enough and such that:
\begin{equation*}
\nu^2 = \left( \frac{q}{N^{1/2 + \epsilon}} \right)^2 < \frac{K\, K_1^2}{N}\, .
\end{equation*}
Then  $\underset{N \longrightarrow \infty}{\lim}  \nu^2  \left\vert a_k(\lambda) \right\vert^2  \, \left\vert b_k \right\vert^2 = 0$ and $\log \left(  1 - \nu^2 \, \left\vert \tilde{a}_k(\lambda) \right\vert^2  \, \left\vert b_k \right\vert^2  \right)^{-1} \sim \nu^2  \left\vert \tilde{a}_k(\lambda) \right\vert^2 \,  \left\vert b_k \right\vert^2 $. Thereby, with $A$ defining a constant, it can be obtained:
\begin{equation*}
\mathds{P} \left[   \hat{\gamma}_m^{cross} (\lambda) > x  \mid \mathbf{T} \right] \leq \exp{\left(-N^{1/2 - \epsilon} \, q \, x - A\right)}\, .
\end{equation*}
Then, integrating with respect to any density $p_\mathbf{T}(.)$ of $\mathbf{T}$ leads to:
\begin{align*}
&\mathds{P} \left[   \hat{\gamma}_m^{cross} (\lambda) > x   \right]  = \int \mathds{P} \left[   \hat{\gamma}_m^{cross} (\lambda) > x  \mid \mathbf{T} \right] \, p_\mathbf{T}(\mathbf{T}) \, d\mathbf{T} \\
& \leq \exp{\left(-N^{1/2 - \epsilon} \, q \, x - A\right)} \, .
\end{align*}

This proves that, for any $\lambda_i$, $\mathds{P} \left[   \hat{\gamma}_m^{cross} (\lambda_i) > x   \right] \underset{N \rightarrow \infty}{\rightarrow} 0$. \\ It remains now to prove that $\underset{i \in \mathcal{I}}{\max} \underset{\lambda \in [\lambda_i \, \, \lambda_{i+1}] }{\sup}   \left\vert \hat{\gamma}_m^{cross} (\lambda) - \hat{\gamma}_m^{cross} (\lambda_i)\right\vert \overset{a.s.}{\longrightarrow} 0$. This will be left to the reader as it follows the same proof as for $\chi_1$ of \eqref{eqgamma}. We have so $\mathds{P} \left[  \underset{\lambda \in [0 \, 2\pi)}{\sup} \hat{\gamma}_m^{cross} (\lambda) > x   \right] \underset{N \longrightarrow \infty}{\longrightarrow} 0 $.


\subsubsection{Analysis of $\underset{\lambda \in \left[  0,2 \pi \right)  }{\sup} \left\vert \hat{\gamma}_m^{sign} (\lambda) \right\vert$}
The proof of convergence of this quantity follows the same principles. We have:
\begin{equation*}
\hat{\gamma}_m^{sign} (\lambda) = \mathbf{d}_m^H (\lambda) \, \frac{\mathbf{M} \, \boldsymbol{\delta}^H \, \boldsymbol{\Gamma} \, \boldsymbol{\delta}\, \mathbf{M}^H}{N} \, \mathbf{d}_ m (\lambda) \, .
\end{equation*}

\noindent As previously,  let $\mathds{I}_{m}$ be a $m \times m$ matrix containing 1 everywhere and let $\mathbf{E}(\lambda)= \mathbf{D}_m(\lambda) \, \mathds{I}_m \, \mathbf{D}_m^H(\lambda)$. Then:
\begin{equation*}
\hat{\gamma}_m^{sign} (\lambda) = 2 \, \mathcal{R}e \left[   \frac{1}{N} \, \mathrm{Tr}\left( \mathbf{M} \, \boldsymbol{\delta} \, \boldsymbol{\Gamma} \, \boldsymbol{\delta}^H  \, \mathbf{M}^H \mathbf{E}\right) \right]\, .
\end{equation*}
\noindent Let  $\mathbf{A}(\lambda) = \mathbf{M}^H \, \mathbf{E} \, \mathbf{M} \, \boldsymbol{\delta}$ and $\mathbf{B} = \boldsymbol{\Gamma}  \, \boldsymbol{\delta}^H$ be two matrix respectively of size $p \times N$ and $N \times p$. Defining $\mathbf{a}(\lambda) = \mathrm{vec}(\mathbf{A}(\lambda))$  and $\mathbf{b}= \mathrm{vec}(\mathbf{B})$, we have:
\begin{align*}
\hat{\gamma}_m^{sign} (\lambda) &= \frac{2}{N} \mathcal{R}e \left(  \mathrm{vec}^T(\mathbf{A}(\lambda))  \, \mathrm{vec}(\mathbf{B})   \right)\, , \\
&= \displaystyle \frac{2}{N} \, \mathcal{R}e \left(   \mathbf{a}^T(\lambda) \, \left(\mathbf{M}^H \, \mathbf{E} \, \mathbf{M} \otimes \mathbf{I}_N\right) ^{-T}  \right. \\
& \left. \left(\mathbf{M}^H \,\mathbf{E} \,\mathbf{M} \otimes \mathbf{I}_N \right)^T \, \left( \boldsymbol{\Gamma} \otimes \mathbf{I}_p  \right) \, \left(\boldsymbol{\Gamma} \otimes \mathbf{I}_p\right)^{-1} \,  \mathbf{j}   \right)\, , \\
&=\displaystyle  \frac{2}{N} \,   \sum_{k=1}^{N \,  p} \mathcal{R}e (a_k(\lambda) ) \,  \mathcal{R}e (b_k) - \mathcal{I}m (a_k(\lambda)) \, \mathcal{I}m (b_k)\, .
\end{align*}

Let us define $ \mathbf{C}(\lambda) = \mathbf{M}^H \, \mathbf{E} \, \mathbf{M} \otimes \mathbf{I}_N$, $\mathbf{D} = \boldsymbol{\Gamma} \otimes \mathbf{I}_p$, $\tilde{\mathbf{a}}(\lambda) =   \mathbf{C}^{-1}(\lambda) \, \mathbf{a}(\lambda)$,  $\tilde {\mathbf{b}} = \mathbf{D}^{-1} \, \mathbf{b}$,  $ c_k = \displaystyle \sum_{l=1}^{N \,  p}  \left(\mathbf{C}^{T}(\lambda)\right)_{l,k}$ and $d_k = \displaystyle \sum_{s=1}^{N \, p}  \left(\mathbf{D}\right)_{s,k}$. Using Lemma \ref{lemma6} and the Markov inequality, it can be shown that, for any fixed $\lambda \in [0 \, \,2\pi  )$ and a constant $\mu$ such that $0 < \mu <  \left( \underset{\lambda \in \left[  0,2 \pi \right)  }{\sup} \Vert \mathbf{C}(\lambda) \Vert \sup \Vert \mathbf{D} \Vert \right)^{-1}$:
\begin{align*}
&\mathds{P} \left[   \hat{\gamma}_m^{sign} (\lambda) > x  \right] \\
&\leq \exp{\left(-N \, \nu \, x - \displaystyle \sum_{k=1}^{N\, p} \log \left(  1 - \mu^2  \, \left\vert c_k(\lambda) \right\vert^2  \, \left\vert d_k \right\vert^2  \right)^{-1}\right) } \, .
\end{align*}
\noindent As the matrix $\mathbf{\Gamma}$ is absolutely summable, then, for all $k$, $\left\vert d_k \right\vert^2 \leq K$ where $K$ is a constant. Now, for all $k$, we have
\begin{align*}
\vert c_k(\lambda) \vert  &= \left\vert \displaystyle \sum_{s=1}^{p}  \left[ \sum_{l=1}^{m} M_{l,s}  \sum_{j=1}^{m} M_{j,k} \, e^{i\,(j-l) \lambda} \right] \right\vert \, ,\\
&\leq  \displaystyle \sum_{s=1}^{p}  \left[ \sum_{l=1}^{m} \left\vert M_{l,s} \right\vert \sum_{j=1}^{m} \left\vert M_{j,k} \right\vert \right] \, .
\end{align*} 

\noindent The columns of $\mathbf{M}$ are absolutely summable. As $p$ is fixed and $p \ll N$, with $K$ a constant, we have $\vert c_k(\lambda) \vert \leq K$. The coefficients of the matrix $\boldsymbol{\Gamma}$ being absolutely summable, for all $k$, we have find a constant $K_1$ such that $\vert d_k \, \vert \leq K_1$ . By defining $w$ as a constant small enough and  $\mu = \displaystyle \frac{w}{\sqrt{N}}$ such that $\mu ^2 \, \vert c_k \, \vert^2 \, \vert d_k \vert^2 \underset{N \longrightarrow \infty}{\longrightarrow} 0 $, then, for all $x > 0$ and $A$ a constant, we have the following inequality:
\begin{equation*}
\mathds{P} [\hat{\gamma}_m^{sign} (\lambda) > x] \leq \exp{\left(- N^{1/2} \, w \, x - A\right)} \, .
\end{equation*}

\noindent As for $\gamma_m^{cross} (\lambda)$, it remains to prove than $\underset{i \in \mathcal{I}}{\max} \underset{\lambda \in [\lambda_i \, \, \lambda_{i+1}] }{\sup}   \vert \hat{\gamma}_m^{sign} (\lambda) - \hat{\gamma}_m^{sign} (\lambda_i)\vert \overset{a.s.}{\longrightarrow} 0$ and this will left to the reader as it is the same as the proof of $\chi_1$. We have proven than $\mathds{P} \left[  \underset{\lambda \in [0 \, 2\pi)}{\sup} \hat{\gamma}_m^{sign} (\lambda) > x   \right] \underset{N \longrightarrow \infty}{\longrightarrow} 0 $. As the right term of \eqref{eqproof3} tends to zero when $N$ is tends to infinity, the proof of Theorem 1 is completed.


\subsection{Proof of Theorem 3 }

\noindent The proof follow the same idea. With the notation $\check{\mathbf{C}}_{FP} = \mathcal{T} (\hat{\mathbf{C}}_{FP})$ where $\mathcal{T}$ is the Toeplitz operator  defined in the introduction, the equation to prove becomes: 
 \begin{equation}
\left\Vert \mathcal{T} (\hat{\mathbf{C}}_{FP}) -  \mathds{E}\left[v(\tau \gamma) \, \tau\right] \,\mathbf{C} \right\Vert \overset{a.s.}{\longrightarrow} 0 \, .
\label{eqprooof2}
\end{equation}

\noindent This equation can be split as: 
 \begin{align*}
&\left\Vert \mathcal{T} \left( \hat{\mathbf{C}}_{FP} \right) -   \mathds{E}\left[v(\tau \gamma) \,\tau\right] \, \mathbf{C} \right\Vert \\
& \leq \left\Vert \mathcal{T} \left(\hat{\mathbf{C}}_{FP}-  \hat{\mathbf{S}}\right) \right\Vert + \left\Vert  \mathcal{T} \left( \hat{\mathbf{S}} \right) -  \mathds{E}\left[v(\tau \,\gamma) \,\tau\right]\,  \mathbf{C} \right\Vert \, .
\end{align*}

\noindent Let us considering the following notations: 
\begin{itemize}
\item [$\bullet$] $\hat{\mathbf{S}}$ the matrix such as $ \left\Vert \mathbf{\check{\Sigma}} - \hat{\mathbf{S}} \right\Vert \overset{a.s.}{\longrightarrow} 0$, as Theorem 3 has stated. As a reminder, $\hat{\mathbf{S}}$ is the matrix defined by: 
\begin{equation*}
\hat{\mathbf{S}} = \frac{1}{N}  \displaystyle \sum_{i=1}^{N}  v\left(\tau_i \,\gamma\right) \, \mathbf{y}_{wi} \, \mathbf{y}_{wi}^{H} \, ,
\end{equation*} where $\gamma$ is the unique solution (if defined) of: 
\begin{equation*}
1 = \frac{1}{N} \displaystyle \sum_{i=1}^{N} \frac{\psi(\tau_i \, \gamma)}{1+ c \,\psi(\tau_i \,\gamma)}\, ,
\end{equation*} 
where $g : x \mapsto \displaystyle \frac{x}{1-c\,\phi(x)}$, $v : x \mapsto u \, o \, g^{-1} (x)$ and $\psi:x\mapsto x \, v(x)$.
\item [$\bullet$] If $\mathbf{A} = \mathcal{T}\left(\left(a_0, \ldots, a_{m-1}\right)^T\right)$ is a Toeplitz matrix ($a_{-k} = a_k^\ast$), we can define the spectral density as: 
\begin{equation*}
\gamma^{\mathbf{A}} (\lambda) \overset{\Delta}{=} \displaystyle \sum_{k=1-m}^{m-1} a_k \,e^{i\,k\,\lambda} \, .
\end{equation*}
Finally, we denote by $\hat{\gamma}^{\mathbf{A}}(\lambda)$ the estimated spectral density of Toeplitz matrix $\mathbf{A}$.
\end{itemize} 

To prove the consistency, we will decompose, as for Theorem 1, the equation \eqref{eqpreuveth21} in two parts. As matrices $\mathcal{T} \left(\hat{\mathbf{C}}_{FP}\right)$ and $\mathbf{C}$ are Toeplitz, it follows through \eqref{eqGray}:
\begin{eqnarray}
&& \left\Vert \mathcal{T} \left(\hat{\mathbf{C}}_{FP}\right) -   \mathds{E}[v(\tau \,\gamma) \, \tau] \, \mathbf{C} \right\Vert \nonumber \\
& \leq & \underset{\lambda \in \left[  0, 2 \,\pi \right)  }{\sup} \left\vert \hat{\gamma}^{\hat{\mathbf{S}}} (\lambda) - \gamma^{\mathds{E}[v(\tau \, \gamma) \,\tau]  \, \mathbf{C}}(\lambda) \right\vert + \left\Vert \mathcal{T} \left( \hat{\mathbf{C}}_{FP} -  \hat{\mathbf{S}} \right) \right\Vert \, \nonumber \,\\
& \leq & \chi_1 + \chi_2  \label{eqpreuveth21} ,
\end{eqnarray}
where
$\chi_1 = \underset{\lambda \in \left[  0,2 \, \pi \right)  }{\sup} \left\vert \hat{\gamma}^{ \hat{\mathbf{S}}} (\lambda) - \gamma^{ \mathds{E}[v(\tau \,\gamma) \,\tau]\,  \mathbf{C}}(\lambda) \right\vert $ 
and $\chi_2 = \left\Vert \mathcal{T} ( \hat{\mathbf{C}}_{FP} -  \hat{\mathbf{S}} ) \right\Vert$.\\

\subsubsection{Part 1: convergence of $\chi_1 = \underset{\lambda \in \left[  0,2 \, \pi \right)  }{\sup} \left\vert \hat{\gamma}^{ \hat{\mathbf{S}}} (\lambda) - \gamma^{ \mathds{E}[v(\tau \,\gamma) \,\tau]\,  \mathbf{C}}(\lambda) \right\vert $}

We will split $\chi_1$ into two sub-terms:
\begin{eqnarray*}
 && \underset{\lambda \in \left[  0, 2 \,\pi \right)  }{\sup} \left\vert \hat{\gamma}^{ \hat{\mathbf{S}}} (\lambda) - \gamma^{ \mathds{E}[v(\tau \,\gamma) \,\tau] \, \mathbf{C}}(\lambda) \right\vert \\
 & \leq  & \underset{\lambda \in \left[  0,2 \,\pi \right)  }{\sup} \left\vert \hat{\gamma}^{ \hat{\mathbf{S}}} (\lambda) - \mathds{E} \left[\hat{\gamma}^{ \hat{\mathbf{S}}} (\lambda)\right] \right\vert \\
 &+&  \underset{\lambda \in \left[  0,2 \,\pi \right)  }{\sup} \left\vert  \mathds{E} \left[\hat{\gamma}^{\hat{\mathbf{S}}} (\lambda)\right]  - \gamma^{ \mathds{E}[v(\tau \,\gamma) \,\tau] \, \mathbf{C}}(\lambda) \right\vert \, , \\
 & \leq &  \chi_{11} + \chi_{12} \, ,
 \end{eqnarray*}
where $\chi_{11} = \underset{\lambda \in \left[  0,2\, \pi \right)  }{\sup} \left\vert \hat{\gamma}^{ \hat{\mathbf{S}}} (\lambda) - \mathds{E} \left[\hat{\gamma}^{ \hat{\mathbf{S}}} (\lambda)\right] \right\vert $ and $\chi_{12} = \underset{\lambda \in \left[  0,2 \, \pi \right)  }{\sup} \left\vert  \mathds{E} \left[\hat{\gamma}^{ \hat{\mathbf{S}}} (\lambda)\right]  - \gamma^{ \mathds{E}\left[v(\tau \,\gamma) \, \tau\right] \, \mathbf{C}}(\lambda) \right\vert$. \\

\subsubsection*{Part 1.1: convergence of $\chi_{11}$}
We will need the following lemma: 
\begin{lemma}
\label{lemma7}
\begin{equation}
\hat{\gamma}^{ \hat{\mathbf{S}}} (\lambda) =  \mathbf{d}_m^H (\lambda) \, \hat{\mathbf{S}} \, \mathbf{d}_m (\lambda)\, ,
\label{eq1th2lemme5}
\end{equation} 
and: 

\begin{equation}
\mathds{E} \left[\hat{\gamma}^{ \hat{\mathbf{S}}} (\lambda) \right] =  \mathds{E}\left[v(\tau \, \gamma) \, \tau\right] \, \mathbf{d}_m^H (\lambda)  \, \mathbf{I}_m \, \mathbf{d}_m (\lambda)\, ,
\label{eq2th1lemme5}
\end{equation} 

where $\mathbf{d}_m(\lambda) =\displaystyle  \frac{1}{\sqrt{m}} \left[1, e^{-i \,\lambda}, \ldots, e^{-i\,(m-1)\,\lambda}\right]^T$.
\end{lemma}
\begin{proof}
This is the same idea than for Lemma~\ref{lemma1}. First, we can write:
\begin{eqnarray}
&& \hat{\gamma}^{\hat{\mathbf{S}}}(\lambda)  =  \displaystyle \sum_{k=1-m}^{m-1}  \check{s}_k \, e^{i\,k \,\lambda} \, , \nonumber \\
& =  & \frac{1}{m\,N}   \displaystyle \sum_{k=1-m}^{m-1} e^{i\,k \,\lambda}  \displaystyle \sum_{j=0}^{m-1} \sum_{n=0}^{N-1} \hat{s}_{j,n} \, \hat{s}^{\star}_{j+k,n} \, \mathds{1}_{0 \leq j+k< m} \, ,\nonumber \\
& = &  \displaystyle \frac{1}{m\,N} \displaystyle \sum_{l,l'=0}^{m-1} e^{-i\,(l'-l) \,\lambda}   \displaystyle \sum_{n=0}^{N-1} \hat{s}_{l,n} \, \hat{s}^{\star}_{l',n}\nonumber = \mathbf{d}_m^H (\lambda)\,  \hat{\mathbf{S}} \,   \mathbf{d}_m (\lambda) \, .
\end{eqnarray}


The first part of the Lemma is then proven. Concerning $\mathds{E} \left[\hat{\gamma}^{\hat{\mathbf{S}}} (\lambda) \right]$, we can define $\mathbf{D}$ as the diagonal matrix containing the $\left\{v(\tau_i \,\gamma)\right\}_{i\in \llbracket 0, N-1\rrbracket}$. We obtain: 
\begin{align*}
\mathds{E} \left[ \hat{\gamma}^{ \hat{\mathbf{S}}} (\lambda) \right] &= \mathbf{d}_m^H (\lambda) \, \mathds{E}\left[\hat{\mathbf{S}}\right] \, \mathbf{d}_m (\lambda) \, , \\ 
&= \mathbf{d}_m^H(\lambda) \, \mathds{E} \left[  \frac{\mathbf{Y}_w \, \mathbf{D}\, \mathbf{Y}_w^H}{N} \right] \, \mathbf{d}_m (\lambda)  \, .
\end{align*}

Then  expliciting each element of $\mathds{E}\left[\mathbf{Y}_w \, \mathbf{D} \, \mathbf{Y}_w^H\right]$ leads to:
\begin{align*}
& \left(  \mathds{E}\left[ \mathbf{Y}_w \, \mathbf{D} \, \mathbf{Y}_w^H \right] \right) _{i,j} =  \mathds{E} \left[ \displaystyle \sum_{n=0}^{N-1} v(\tau_n \,\gamma) \, y_{w \,i,n} \, y^{\star}_{w \, j,n}  \right]\, , \\ 
&=    \displaystyle \sum_{n=0}^{N-1}  \mathds{E} \left[ v(\tau_n \, \gamma) \, \tau_n\right] = N \, \mathds{E}\left[ v(\tau_n \, \gamma) \, \tau_n\right]\, .
\end{align*}

We obtain the following result: $\mathds{E}\left[\hat{\gamma}^{\hat{\mathbf{S}}} (\lambda)\right] = \mathds{E} \left[v(\tau \,\gamma) \,\tau\right] \, \mathbf{d}_m^H(\lambda) \,\mathbf{I}_m \, \mathbf{d}_m(\lambda)$.

\end{proof}

The rest of the proof for $\chi_{11}$ is the same as for Theorem 1 $\hat{\gamma}^{noise}$, but with $\mathbf{T}$ containing the $\left\{\tau_i\right\}_{i}$ on its diagonal, we will have $\left\Vert  \mathbf{\mathbf{T}} \right\Vert_{\infty} \, \left\Vert  \mathbf{\mathbf{D}} \right\Vert_{\infty}$ instead of $\left\Vert \mathbf{\mathbf{T}} \right\Vert_{\infty}$. We obtain so $\chi_{11} \overset{a.s.}{\longrightarrow} 0$  as $m \longrightarrow \infty$. \\

\subsubsection*{Part 1.2: convergence of $\chi_{12}$} 
Lemma \ref{lemma7} and \eqref{eq1th2lemme5} give us 
\begin{equation*}
\mathds{E}\left[\hat{\gamma}^{\hat{\mathbf{S}}} (\lambda)\right]= \mathds{E} \left[v(\tau \,\gamma) \,\tau \right] \, \mathbf{d}_m^H(\lambda) \, \mathbf{C} \, \mathbf{d}_m (\lambda) \, .
\end{equation*}
\noindent and
$\mathds{E}\left[v(\tau \,\gamma) \,\tau\right] \, \gamma^\mathbf{C}(\lambda) =  \mathds{E} \left[v(\tau \,\gamma) \,\tau\right] \, \mathbf{d}_m^H(\lambda) \, \mathbf{C} \, \mathbf{d}_m (\lambda)$. This yields $\chi_{12} = 0$.\\

\subsubsection{Part 2: convergence of $\chi_2 = \left\Vert \mathcal{T} \left( \hat{\mathbf{C}}_{FP} -  \hat{\mathbf{S}}\right) \right\Vert$}
It is proven, in \cite{Couillet15a} that $\left\Vert  \hat{\mathbf{C}}_{FP} - \hat{\mathbf{S}} \right\Vert \overset{a.s.}{\longrightarrow} 0$. Let $\mathbf{J}$ be a matrix such that  $\left(\mathbf{J}\right)_{j-i = 1} = 1$ and $0$ elsewhere. $\mathbf{J}^k$ contains 1 only on the $k^{th}$ diagonal. As before, thanks to \eqref{eqGray}, we have:
\begin{equation*}
\left\Vert \mathcal{T} \left( \hat{\mathbf{C}}_{FP} -  \hat{\mathbf{S}} \right) \right\Vert  \leq  \underset{\lambda \in \left[  0,2 \pi \right)  }{\sup} \left\vert \displaystyle \sum_{k=1-m}^{m-1} \left( \check{fp}_k - \check{s}_k \right) \, e^{i\,k\, \lambda} \right\vert \, .
\end{equation*}
 Let us define $\mathcal{T} \left( \hat{\mathbf{C}}_{FP} \right) = \mathcal{L}\left(\left(\check{fp}_0, \ldots, \check{fp}_{m-1}\right)^T\right)$ with $\check{fp}_{-k} = \check{fp}_k^\ast$ and $\mathcal{T} \left( \hat{\mathbf{S}}\right) = \mathcal{L}\left(\left(\check{s}_0, \ldots, \check{s}_{m-1}\right)^T\right)$ with $\check{s}_{-k} = \check{s}_k^\ast$. We have:
 \begin{align*}
& \underset{\lambda \in \left[  0,2 \,\pi \right)  }{\sup} \left\vert \displaystyle \sum_{k=1-m}^{m-1} \left( \check{fp}_k - \check{s}_k \right) \, e^{i\, k\, \lambda} \right\vert \\
 &=  \underset{\lambda \in \left[  0, 2 \,\pi \right)  }{\sup} \left\vert \displaystyle \sum_{k=1-m}^{m-1} \frac{1}{m} \sum_{p-1}^m \left(\check{fp}_k - \check{s}_k \right) \, e^{i\, k \, \lambda} \, \mathds{1}_{0 \leq p+k \leq m} \right\vert \, ,
 \\ &=  \underset{\lambda \in \left[  0,2 \,\pi \right)  }{\sup} \left\vert \mathrm{Tr} \left( \left(\hat{\mathbf{C}}_{FP} -  \hat{\mathbf{S}}\right) \, \frac{1}{m} \displaystyle \sum_{k=1-m}^{m-1}   \left(\mathbf{J}^T\right)^k \, e^{i \, k\, \lambda} \right) \right\vert \, .
 \end{align*}


Moreover $\displaystyle \frac{1}{m} \displaystyle \sum_{k=1-m}^{m-1}  \left(\mathbf{J}^T\right)^k \, e^{i \, k\, \lambda} = \mathbf{d}_m (\lambda) \, \mathbf{d}_m^{H} (\lambda)$. This leads to:  
\begin{align*}
&\left\Vert \mathcal{T} \left(\hat{\mathbf{C}}_{FP} -  \hat{\mathbf{S}} \right) \right\Vert  \\
&\leq  \underset{\lambda \in \left[  0, 2 \,\pi \right)  }{\sup} \left\vert \mbox{Tr} \left( \left(\hat{\mathbf{C}}_{FP} -  \hat{\mathbf{S}}\right) \, \mathbf{d}_m(\lambda) \, \mathbf{d}_m^{H} (\lambda) \right) \right\vert \\
&=  \underset{\lambda \in \left[  0, 2 \, \pi \right)  }{\sup} \left\vert \mathbf{d}_m^{H} (\lambda)  \, \left(\hat{\mathbf{C}}_{FP} -  \hat{\mathbf{S}} \right) \, \mathbf{d}_m(\lambda) \right\vert \, .
\end{align*}

For any vector $\mathbf{x}$, the last equation becomes:
\begin{eqnarray*}
&& \underset{\lambda \in \left[  0,2 \,\pi \right)  }{\sup} \left\vert \mathbf{d}_m^{H} (\lambda)  \, \left(\hat{\mathbf{C}}_{FP} -  \hat{\mathbf{S}} \right) \, \mathbf{d}_m(\lambda) \right\vert \\
& \leq & \underset{\left\Vert \mathbf{x} \right\Vert_2  = 1}{\sup} \left\vert \mathbf{x}^H \, \left( \hat{\mathbf{C}}_{FP} -  \hat{\mathbf{S}}\right) \, \mathbf{x} \right\vert \, , \\
& \leq  &\underset{\left\Vert \mathbf{x} \right\Vert_2  = 1}{\sup} \left\Vert \left( \hat{\mathbf{C}}_{FP} -  \hat{\mathbf{S}}\right) \, \mathbf{x} \right\Vert_2  \leq  \left\Vert \hat{\mathbf{C}}_{FP} -  \hat{\mathbf{S}}\right\Vert \, .
\end{eqnarray*}


Finally, we obtain:
\begin{equation*}
\left\Vert \mathcal{T} \left( \hat{\mathbf{C}}_{FP} -  \hat{\mathbf{S}} \right) \right\Vert  \leq \left\Vert \hat{\mathbf{C}}_{FP} -  \hat{\mathbf{S}} \right\Vert \, .
\end{equation*}

As $\left\Vert \hat{\mathbf{C}}_{FP} - \hat{\mathbf{S}} \right\Vert   \overset{a.s.}{\longrightarrow} 0$ then $\chi_2 \overset{a.s.}{\longrightarrow} 0$ and  the proof of Theorem 3 is completed.

\section{Proof of Theorem 2}

As the proof is the same for $\check{\mathbf{\Sigma}}_{SCM}$ and $\check{\mathbf{\Sigma}}_{FP}$, let $\check{\mathbf{\Sigma}}$ denote one or the other of these matrices. \\

\noindent From the equations \eqref{eqsigma1} and \eqref{eqsigma2}, as $\mathbf{\check{y}}_{wi} = \mathbf{\check{C}}^{-1/2} \, \mathbf{y}_i$, $\check{\mathbf{\Sigma}}$ is the unique solution of:
\begin{align*}
\mathbf{\Sigma} = \frac{1}{N} \displaystyle \sum_{i=0}^{N-1}  & u \left( \frac{1}{m} \, \mathbf{y}_i^H  \, \mathbf{\check{C}}^{-1/2} \,  \mathbf{\Sigma}^{-1} \, \mathbf{\check{C}}^{-1/2} \,\mathbf{y}_i \right)\, \times \\
& \mathbf{\check{C}}^{-1/2} \,\mathbf{y}_i \, \mathbf{y}_i^H \, \mathbf{\check{C}}^{-1/2} \, .
\end{align*}

Rewriting this equation with the $\left\{\mathbf{\check{y}}_{wi}\right\}_i$ 
\begin{align*}
 & \mathbf{C}^{-1/2} \mathbf{\check{C}}^{1/2} \, \mathbf{\Sigma} \, \mathbf{\check{C}}^{1/2} \, \mathbf{C}^{-1/2} \\ 
 & = \displaystyle \frac{1}{N}  \, \sum_{i=0}^{N-1}  u \left( \frac{1}{m} \, \mathbf{\check{y}}_{wi}^H \, \left(   \mathbf{C}^{-1/2} \, \mathbf{\check{C}}^{1/2} \, \mathbf{\Sigma}\, \mathbf{\check{C}}^{1/2} \, \mathbf{C}^{-1/2}  \,  \right)^{-1} \right. \times \\
 & \left. \mathbf{\check{y}}_{wi}  \right) \, \mathbf{\check{y}}_{wi}   \, \mathbf{\check{y}}_{wi}^H \, ,
\end{align*}
we obtain the following relationship between $\check{\mathbf{\Sigma}}$ and $\hat{\mathbf{\Sigma}}$:
\begin{equation}
\check{\mathbf{\Sigma}} = \mathbf{\check{C}}^{-1/2} \, \mathbf{C}^{1/2} \, \hat{\mathbf{\Sigma}} \, \mathbf{C}^{1/2} \, \mathbf{\check{C}}^{-1/2} \, .
\label{Sigmacheck}
\end{equation}

Then, equation \eqref{eqcouilletmod} can be rewritten as
\begin{equation}
\left\Vert \check{\mathbf{\Sigma}} -  \hat{\mathbf{S}}\right\Vert \leq \left\Vert \check{\mathbf{\Sigma}} -  \hat{\mathbf{\Sigma}}\right\Vert  \, + \left\Vert \hat{\mathbf{\Sigma}} -  \hat{\mathbf{S}}\right\Vert\, .
\label{inequality}
\end{equation}

Concerning the second term of the right hand side of \eqref{inequality}, it is proven in \cite{Couillet15b} that the matrix $\hat{\mathbf{S}}$ given by \eqref{matrixS} is such that
\begin{equation}
\left\Vert \hat{\mathbf{\Sigma}} - \hat{\mathbf{S}} \right\Vert \overset{a.s.}{\longrightarrow} 0 \, .
\label{eqcouillet}
\end{equation}

With \eqref{Sigmacheck}, the first term of right hand side of \eqref{inequality} can be rewritten as: 
\begin{eqnarray}
&& \left\Vert \check{\mathbf{\Sigma}}- \mathbf{\hat{\Sigma}} \right\Vert  \le  \left\Vert \mathbf{\check{C}}^{-1/2} \, \mathbf{C}^{1/2}    \, \mathbf{\hat{\Sigma}}   \, \mathbf{C}^{1/2} \, \mathbf{\check{C}}^{-1/2} -  \mathbf{\hat{\Sigma}}   \, \mathbf{C}^{1/2} \, \mathbf{\check{C}}^{-1/2}\right\Vert \nonumber \\ 
&& +   \left\Vert \mathbf{\hat{\Sigma}}   \, \mathbf{C}^{1/2} \, \mathbf{\check{C}}^{-1/2} - \mathbf{\hat{\Sigma}} \right\Vert \, .
\end{eqnarray}
After left and right factorizations, we obtain:
\begin{eqnarray}
\left\Vert \check{\mathbf{\Sigma}}- \mathbf{\hat{\Sigma}} \right\Vert \le \left\Vert \mathbf{\check{C}}^{-1/2} \, \mathbf{C}^{1/2}  -\mathbf{I}_m\right \Vert  \left\Vert  \mathbf{\hat{\Sigma}}  \right\Vert  \left(\left\Vert \mathbf{C}^{1/2} \, \mathbf{\check{C}}^{-1/2} \right\Vert +1\right) \, . \nonumber
\end{eqnarray}
 As $\Vert \mathbf{C}\Vert $ has a bounded support, $\Vert  \check{\mathbf{C}} \Vert $ is bounded too since  its eigenvalues support converges almost surely toward the true distribution. Moreover, Theorem 1 and Theorem 2 have proved the consistency $ \left\Vert \mathbf{C} -\check{\mathbf{C}} \right\Vert \overset{a.s.}{\longrightarrow} 0$. This ensures the proof.

\end{appendices}

\section*{Acknowledgment}
The authors would like to thank the DGA for its financial support.

\bibliographystyle{IEEEtran}

\bibliography{bibliographieeug}

\end{document}